\documentclass[10pt,journal,compsoc]{IEEEtran}
\ifCLASSOPTIONcompsoc
  \usepackage[nocompress]{cite}
\else
  \usepackage{cite}
\fi
\ifCLASSINFOpdf
\else
\fi

\usepackage{caption}
\usepackage{algorithm}
\usepackage{algpseudocode}

\usepackage{color}
\usepackage{amsthm}
\usepackage{amsmath}
\usepackage{verbatim}
\usepackage{graphicx}
\usepackage{algorithm} \usepackage{algpseudocode} \usepackage{amsmath}
\usepackage{lipsum}
\usepackage{subfigure}

\usepackage{multirow}
\usepackage[top=2cm, bottom=2cm, left=2cm, right=2cm]{geometry}
\usepackage{algorithmicx}
\usepackage{algpseudocode}
\usepackage{float}

\begin{document}
%
\title{BrePartition: Optimized High-Dimensional \emph{k}NN Search with Bregman Distances}
%
%
%
%

\author{Yang~Song, Yu~Gu, Rui~Zhang,~\IEEEmembership{Senior Member,~IEEE} and Ge~Yu,~\IEEEmembership{Member,~IEEE}
\IEEEcompsocitemizethanks{\IEEEcompsocthanksitem Y. Song, Y. Gu, and G. Yu are with the School of Computer Science and Engineering, Northeastern University, Shenyang, Liaoning 110819, China. E-mail: ysqyw1994@163.com, \{guyu, yuge\}@mail.neu.edu.cn.
\IEEEcompsocthanksitem R. Zhang is with the School of Computing and Information Systems, The University of Melbourne, Parkville VIC 3010, Australia. E-mail: rui.zhang@unimelb.edu.au.
\IEEEcompsocthanksitem Corresponding author: G. Yu}}
\IEEEtitleabstractindextext{%

\begin{abstract}
Bregman distances (also known as Bregman divergences) are widely used in machine learning, speech recognition and signal processing, and \emph{k}NN searches with Bregman distances have become increasingly important with the rapid advances of multimedia applications. Data in multimedia applications such as images and videos are commonly transformed into space of hundreds of dimensions. Such high-dimensional space has posed significant challenges for existing \emph{k}NN search algorithms with Bregman distances, which could only handle data of medium dimensionality (typically less than 100). This paper addresses the urgent problem of high-dimensional \emph{k}NN search with Bregman distances. We propose a novel partition-filter-refinement framework. Specifically, we propose an optimized dimensionality partitioning scheme to solve several non-trivial issues. First, an effective bound from each partitioned subspace to obtain exact \emph{k}NN results is derived. Second, we conduct an in-depth analysis of the optimized number of partitions and devise an effective strategy for partitioning. Third, we design an efficient integrated index structure for all the subspaces together to accelerate the search processing. Moreover, we extend our exact solution to an approximate version by a trade-off between the accuracy and efficiency. Experimental results on four real-world datasets and two synthetic datasets show the clear advantage of our method in comparison to state-of-the-art algorithms.
\end{abstract}
}
\maketitle
\IEEEdisplaynontitleabstractindextext
\IEEEpeerreviewmaketitle
\IEEEraisesectionheading{\section{Introduction}}
\IEEEPARstart{B}{regman} distances (also called Bregman divergences), as a generalization of a wide range of non-metric distance functions (e.g., Squared Mahalanobis Distance and Itakura-Saito Distance), play an important role in many multimedia applications such as image and video analysis and retrieval, speech recognition and time series analysis~\cite{DBLP:conf/iccv/GoldbergerGG03, DBLP:conf/compgeom/NielsenN06, DBLP:conf/aaai/LongZY07}. This is because metric measurements (such as Euclidian distance) satisfy the basic properties in metric space, such as non-negativity, symmetry and triangular inequality. Although it is empirically proved successful, the metric measurements represented by Euclidian distance are actually inconsistent with human's perception of similarity~\cite{DBLP:conf/aaai/MuY10, DBLP:conf/nips/LaubMMW06}. Examples from~\cite{DBLP:journals/pvldb/ZhangOPT09} and~\cite{DBLP:conf/aaai/MuY10} illustrate that the distance measurement is not metric when comparing images. As can be seen in Fig.~\ref{Examples}(a), the moon and the apple are similar in shape, the pentagram and the apple are similar in color, but there is no similarity between the moon and the pentagram. In this case, our perception of similarity violates the notion of triangular inequality and illustrates that human beings are often comfortable when deploying or using non-metric dissimilarity measurements instead of metric ones especially on complex data types~\cite{DBLP:journals/pvldb/ZhangOPT09, DBLP:conf/iccv/PuzichaRTB99}. Likewise, as shown in Fig~\ref{Examples}(b)~\cite{DBLP:conf/aaai/MuY10}, both the ¡°man¡± and the ¡°horse¡± are perceptually similar to their composition, but the two obviously differ from each other. Therefore, it is not appropriate to employ Euclidian distance as the distance measurement in many practical scenarios.

Since Bregman distances have the capability of exploring the implicit correlations of data features~\cite{DBLP:journals/tkde/LiPZZWLY17}, they have been widely used in recent decades in a variety of applications, including image retrieval, image classification and sound processing~\cite{DBLP:journals/tmm/RasiwasiaMV07, DBLP:conf/iccv/PuzichaRTB99, gray1980distortion}. Over the last several years, they are also used in many practical applications. For example, they are employed to measure the closeness between Hermitian Positive-Definite (HPD) matrices to realize target detection in a clutter~\cite{DBLP:journals/dsp/HuaCWQC18}. They are also used as the similarity measurements of the registration functional to combine various types of image characterization as well as spatial and statistical information in image registration~\cite{DBLP:journals/pr/FerreiraRB18} and apply multi-region information to express the global information in image segmentation~\cite{DBLP:journals/mta/ChengSTL19}. In addition, they are applied in graph embedding, matrix factorization and tensor factorization in the field of social network analysis~\cite{DBLP:journals/corr/abs-1908-02573}. Among the operations that employ Bregman distances, \emph{k}NN queries are demanded as a core primitive or a fundamental step in the aforementioned applications~\cite{DBLP:conf/nips/WeinbergerBS05, DBLP:conf/icip/RamaswamyR08, DBLP:journals/jip/DongIX16}.
\begin{figure}
\centering
\hspace{-10pt}
\subfigure[Example 1]{
\includegraphics[height=1in]{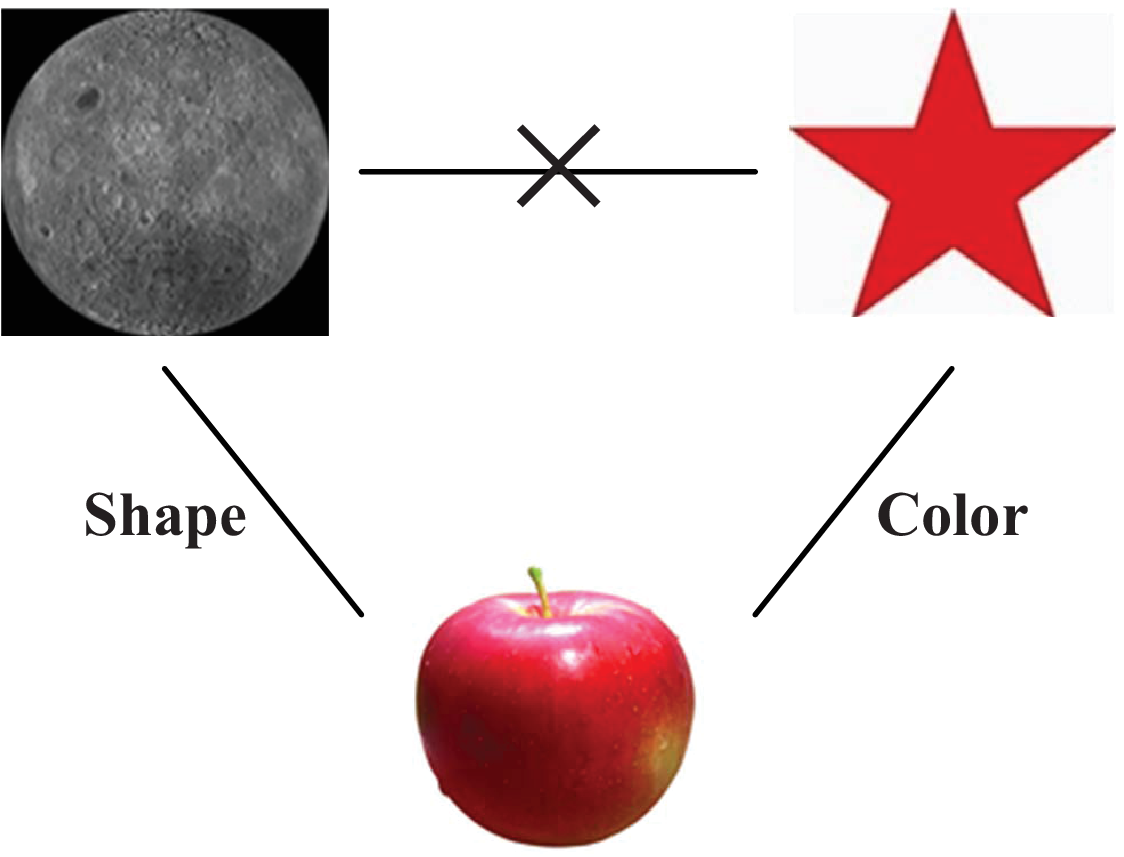}}
\hspace{-10pt}
\subfigure[Example 2]{
\includegraphics[height=1in]{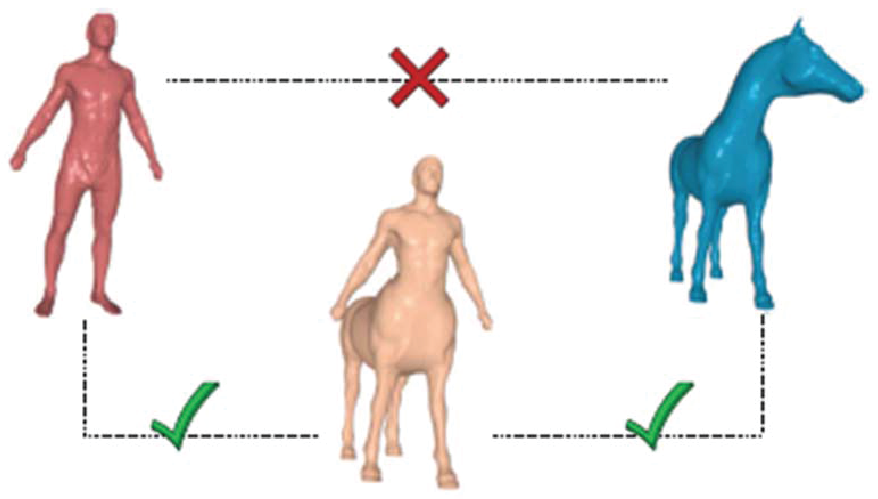}}
\caption{Examples}
\label{Examples}
\end{figure}

In addition, data in multimedia applications such as images and videos are commonly transformed into space of hundreds of dimensions, such as the commonly used datasets (Audio, Fonts, Deep, SIFT, etc) illustrated in our experimental evaluation. Existing approaches~\cite{DBLP:conf/icml/Cayton08, DBLP:journals/pvldb/ZhangOPT09} for \emph{k}NN searches with Bregman distances focus on designing index structures for Bregman distances, but these index structures perform poorly in high-dimensional space due to either large overlaps between clusters or intensive computation.

Aiming at these situations, this paper addresses the urgent problem of high-dimensional \emph{k}NN searches with Bregman distances. In this paper, we propose a partition-filter-refinement framework. We first partition a high-dimensional space into many low-dimensional subspaces. Then, range queries in all the subspaces are performed to generate candidates. Finally, the exact \emph{k}NN results are evaluated by refinement from the candidates. However, realizing such a framework requires solving the following three challenges:
\begin{itemize}
\item \textbf{Bound:} In metric space, bounds are usually derived based on triangular inequality, but in non-metric space, triangular inequality does not hold. Therefore, it is challenging to derive an efficient bound for Bregman distances which are not metric.
\item \textbf{Partition:} How to partition the dimensions to get the best efficiency is a challenge. We need to work out both how many partitions and at which dimensions should we partition.
\item \textbf{Index:} It is challenging to design an I/O efficient index that handles all the dimension partitions in a unified manner. Specifically, by effectively organizing the data points on disks, the index is desired to adapt to our proposed partition strategy and facilitate the data's reusability across partitions.
\end{itemize}

To address the above challenges, and make the following contributions in this paper:

\begin{itemize}
\item We derive the upper bounds between the given query point and an arbitrary data point in each subspace mathematically based on Cauchy inequality and the proper upper bounds are selected as the searching bounds from these subspaces. The final candidate set is the union of the candidate subsets of all the subspaces. We theoretically prove that the \emph{k}NN candidate set obtained following our bound contains the \emph{k}NN results.
\item For dimensionality partitioning, we observe a trade-off between the number of partitions and the search time. Therefore, we derive the algorithm's time complexity and the optimized number of partitions in theory. Furthermore, a strategy named Pearson Correlation Coefficient-based Partition (PCCP) is proposed to reduce the size of the candidate set by partitioning highly-correlated dimensions into different subspaces.
\item After dimensionality partitioning, we employ BB-trees in our partitioned low-dimensional subspaces and design an integrated and disk-resident index structure, named BB-forest. BB-trees can handle low-dimensional data efficiently, so they work well in our framework since the data has been partitioned into low-dimensional subspaces. In addition, the data points are well-organized on the disks based on our proposed PCCP to improve data's reusability in all the subspaces, so that the I/O cost can be reduced.
\item In order to improve the search efficiency while ensuring comparable accuracy with the probability guarantee, we make a trade-off between the efficiency and the accuracy and propose a solution to approximate $k$NN search through the data distribution.
\item Extensive experimental evaluations demonstrate the high efficiency of our approach. Our algorithm named BrePartition can clearly outperform state-of-the-art algorithms in running time and I/O cost.
\end{itemize}

\par The rest of the paper is structured as follows. Section~\ref{Related works} presents the related works. The preliminaries and overview are discussed in Section~\ref{Preliminaries}. We present the derivation of the upper bound in Section~\ref{Derivation of Bound} and the dimensionality partitioning scheme in Section~\ref{Dimensionality Partition}. The index structure, BB-Forest, is described in Section~\ref{BB-Forest}. We present the overall framework in Section~\ref{Putting Everything Together}. The extended solution to approximate \textit{k}NN search is presented in Section~\ref{Extension}. Experimental results are disscussed in Section~\ref{Experiments}. Finally, we conclude our work in Section~\ref{Conclusion}.
\vspace{-10pt}
\section{Related Works}\label{Related works}
\emph{k}NN search is a fundamental problem in many application domains. Here we review existing works on the \emph{k}NN search problem in both metric and non-metric spaces.
\vspace{-10pt}
\subsection{Metric Similarity Search}
The metric similarity search problem is a classic topic and a plethora of methods exist for speeding up the nearest neighbor retrieval. Existing methods contain tree-based data structures including KD-tree~\cite{DBLP:journals/cacm/Bentley75}, R-tree~\cite{DBLP:conf/sigmod/Guttman84}, B+-tree variations~\cite{DBLP:journals/tods/JagadishOTYZ05, DBLP:journals/pvldb/AroraSK018, DBLP:conf/icde/ZhangOT04} and transformation-based methods including Product Quantization (PQ)-based methods~\cite{DBLP:journals/pami/HeoLY19, DBLP:conf/cikm/LiuCC17}, Locality Sensitive Hashing (LSH) family~\cite{DBLP:journals/pvldb/LiuCHLS14, DBLP:conf/sigmod/GaoJLO14, DBLP:journals/pvldb/ChenGZJYY17, DBLP:conf/icde/LiuWZWQ19} and some other similarity search methods based on variant data embedding or dimensionality reduction techniques~\cite{DBLP:conf/cvpr/HwangHA12, DBLP:journals/pvldb/SunWQZL14, DBLP:journals/tkde/GuGSZY18}. These methods can't be utilized in non-metric space where the metric postulates, such as symmetry and triangle inequality, are not followed.
\vspace{-10pt}
\subsection{Bregman Distance-based Similarity Search}
Due to the widespread use of Bregman distances in multimedia applications, a growing body of work is tailored for the \emph{k}NN search with bregman distances. The prime technique is Bregman voronoi diagrams derived by Nielsen et al.~\cite{DBLP:conf/soda/NielsenBN07}. Soon after that, Bregman Ball tree (BB-tree) is introduced by Cayton~\cite{DBLP:conf/icml/Cayton08}. BB-trees are built by a hierarchical space decomposition via \emph{k}-means, sharing the similar flavor with KD-tree. In a BB-tree, the clusters appear in terms of Bregman balls and the filtering condition in the dual space on the Bregman distance from a query to a Bregman ball is computed for pruning out portions of the search space. Nielsen  et al.~\cite{DBLP:conf/icml/CovielloMCL13} extend the BB-tree to symmetrized Bregman distances. Cayton~\cite{DBLP:conf/nips/Cayton09} explores an algorithm to solve the range query based on BB-tree. Nevertheless, facing higher dimensions, considerable overlap between clusters will be incurred, and too many nodes have to be traversed during a \emph{k}NN search. Therefore the efficiency of BB-tree is dramatically degraded, sometimes even worse than the linear search. Zhang et al.~\cite{DBLP:journals/pvldb/ZhangOPT09} devise a novel solution to handle the class of Bregman distances by expanding and mapping data points in the original space to a new extended space. It employs typical index structures, R-tree and VA-file, to solve exact similarity search problems. But it's also inefficient for more than 100 dimensions, because too many returned candidates in a filter-refinement model lead to intensive computation overhead of Bregman distances.

Towards more efficient search processing, there have been increasing attentions focusing on the approximate search methods for Bregman distances~\cite{DBLP:conf/soda/AckermannB09, DBLP:conf/icml/Cayton08, DBLP:journals/ijcga/AbdullahMV13, DBLP:conf/icml/CovielloMCL13, DBLP:conf/sac/FerreiraRB17, DBLP:conf/aaai/MuY10}. These approximate methods achieve the efficiency promotions with the price of losing accuracies. For example, the state-of-the-art approximate solution~\cite{DBLP:conf/icml/CovielloMCL13}, which is designed for the high-dimensional space, exploits the data¡¯s distribution and employs a variational approximation to estimate the explored nodes during backtracking in the BB-tree. Nevertheless, all these methods can't provide the precision guarantees, while some of them cannot be applied to the high-dimensional space~\cite{DBLP:conf/icml/Cayton08, DBLP:conf/soda/AckermannB09, DBLP:journals/ijcga/AbdullahMV13, DBLP:conf/sac/FerreiraRB17}.
\vspace{-10pt}
\subsection{Other Non-Metric Similarity Search Methods}
There also exist many works in the context of non-metric similarity search without the limit to Bregman distances. Space-embedding techniques~\cite{DBLP:conf/sigmod/FaloutsosL95, DBLP:journals/kais/WangWLSSZ00, DBLP:conf/cvpr/AthitsosASK04, DBLP:conf/sigmod/AthitsosHKS05} embed non-metric spaces into an Euclidean one where two points that are close to each other in the original space are more likely close to each other in the new space. Distance-mapping techniques transform the non-metric distance by modifying the distance function itself while preserving the original distance orderings. Skopal~\cite{DBLP:journals/tods/Skopal07} develops TriGen algorithm to derive an efficient mapping function among concave functions by using the distance distribution of the database. NM-tree~\cite{DBLP:conf/dexa/SkopalL08} combines M-tree and TriGen algorithm for the non-metric search. Liu et al,~\cite{DBLP:conf/mm/LiuH09} propose a simulated-annealing-based technique to derive optimized transform functions while preserving the original similarity orderings. Chen et al.~\cite{DBLP:journals/tkde/ChenL08} employ the constant shifting embedding with a suitable clustering of the dataset for a more effective lower-bounds. Recently, a representative technique based on Distance-Based Hashing (DBH)~\cite{DBLP:conf/sdm/JangyodsukPA15} is presented and a general similarity indexing methods for non-metric distance measurements is designed by optimizing hashing functions. In addition, Dyndex~\cite{DBLP:conf/mm/GohLC02}, as the most impressive technique based on classification performs classification of the query point to answer similarity search by categorizing points into classes. These methods degrade dramatically in performance when dealing with high-dimensional issues. There exists an approximate solution called Navigable Small World graph with controllable Hierarchy (HNSW)~\cite{DBLP:journals/corr/MalkovY16}, which can be extended to non-metric space. However, it is not a disk-resient solution, while we mainly focus on disk-resident solutions in this paper.

We summarize the properties of representative non-metric search methods in Table~\ref{Summary}. In Table~\ref{Summary}, NM means that the method adopts the distance functions of non-metric space instead of metric space, BDS means that the method is designed specifically for Bregman distances, and HD means that the method works well in high-dimensional space (more than 100 dimensions). Our proposed solution BrePartition is the first algorithm that possesses all the four desired properties compared to existing algorithms.

\begin{table}
\centering
\caption{Non-metric search methods}
\label{Summary}
\begin{tabular}{lllll}
\hline\noalign{\smallskip}
Name & NM & BDS & HD & Exact \\
\noalign{\smallskip}\hline\noalign{\smallskip}
BrePartition (Our solution) & $\surd$ & $\surd$ & $\surd$ & $\surd$ \\
Zhang et al.~\cite{DBLP:journals/pvldb/ZhangOPT09} & $\surd$ & $\surd$ &  & $\surd$ \\
BB-tree~\cite{DBLP:conf/icml/Cayton08, DBLP:conf/nips/Cayton09} & $\surd$ & $\surd$ &  & $\surd$ \\
Bregman voronoi diagram~\cite{DBLP:conf/soda/NielsenBN07} & $\surd$ & $\surd$ &  & $\surd$ \\
BB-tree variants~\cite{DBLP:conf/icml/CovielloMCL13} & $\surd$ & $\surd$ &  & \\
Ackermann et al.~\cite{DBLP:conf/soda/AckermannB09} & $\surd$ & $\surd$ &  &  \\
Abdullah et al.~\cite{DBLP:journals/ijcga/AbdullahMV13} & $\surd$ & $\surd$ &  &  \\
Coviello et al.~\cite{DBLP:conf/icml/CovielloMCL13} & $\surd$ & $\surd$ & $\surd$ &  \\
Ferreira et al.~\cite{DBLP:conf/sac/FerreiraRB17} & $\surd$ & $\surd$ &  &  \\
Non-metric LSH~\cite{DBLP:conf/aaai/MuY10} & $\surd$ & $\surd$ & $\surd$ & \\
FastMap~\cite{DBLP:conf/sigmod/FaloutsosL95} & $\surd$ &  &  & $\surd$ \\
Wang et al.~\cite{DBLP:journals/kais/WangWLSSZ00} & $\surd$ &  &  & $\surd$ \\
Boostmap~\cite{DBLP:conf/cvpr/AthitsosASK04} & $\surd$ &  &  &  \\
Athitsos et al.~\cite{DBLP:conf/sigmod/AthitsosHKS05} & $\surd$ &  & $\surd$ &  \\
TriGen~\cite{DBLP:journals/tods/Skopal07} & $\surd$ &  &  & $\surd$ \\
NM-tree~\cite{DBLP:conf/dexa/SkopalL08} & $\surd$ &  &  &  \\
Liu et al~\cite{DBLP:conf/mm/LiuH09} & $\surd$ &  &  &  \\
LCE~\cite{DBLP:journals/tkde/ChenL08} & $\surd$ &  &  &  \\
DBH~\cite{DBLP:conf/sdm/JangyodsukPA15} & $\surd$ &  & $\surd$ & \\
Dyndex~\cite{DBLP:conf/mm/GohLC02} & $\surd$ &  & $\surd$ & \\
HNSW~\cite{DBLP:journals/corr/MalkovY16} & $\surd$ &  & $\surd$ & \\
\noalign{\smallskip}\hline
\end{tabular}
\end{table}
\vspace{-10pt}
\section{Preliminaries and Overview}\label{Preliminaries}
\subsection{Bregman Distance}
Given a $d$-dimensional vector space $S$, a query $y=(y_1,y_2,...,y_d)$ and an arbitrary data point $x=(x_1,x_2,...,x_d)$, the Bregman distance between $x$ and $y$ is defined as ${D}_{f}\left(x,y\right)=f\left(x\right)-f\left(y\right)-\left<\nabla f\left(y\right),x-y\right>$, where $f(\cdot)$ is a convex function mapping points in $S$ to real numbers, $\nabla f\left(y\right)$ is the gradient of $f(\cdot)$ at $y$, and $\left<\cdot,\cdot\right>$ denotes the dot product between two vectors. When different convex functions are employed, Bregman distances define several well-known distance functions. Some representatives contain:
\begin{itemize}
\item\textbf{Squared Mahalanobis Distance:} The given $f\left(x\right)=\frac{1}{2}{x}^{\rm T}Qx$ generates ${D}_{f}\left(x,y\right)=\frac{1}{2}{\left(x-y\right) }^{T}Q\left(x-y\right)$ which can be considered as a generalization of the above squared Euclidean distance.
\item\textbf{Itakura-Saito Distance (ISD):} When the given $f\left(x\right)=-\sum{\log{{x}_{i}}}$, the distance is IS distance which is denoted by ${D}_{f}\left(x,y\right)=\sum{\left(\frac{{x}_{i}}{{y}_{i}}-\log{\frac{{x}_{i}}{{y}_{i}}-1}\right)}$.
\item\textbf{Exponential Distance (ED):} When the given $f(x)=e^{x}$, the Bregman distance is represented as ${D}_{f}(x,y)=e^{x}-(x-y+1)e^{y}$. In this paper, we name it \emph{exponential distance}.
\end{itemize}

In addition, our method can be applied to most measures belonging to Bregman distances, such as Shannon entropy, Burg entropy, $l_{p}$-quasi-norm and $l_{p}$-norm, except KL-divergence, since it's not cumulative after the dimensionality partitioning.
\vspace{-10pt}
\subsection{Overview}
Our method consists of the precomputation and the search processing. In the precomputation, we first partition the dimensions (described in Section~\ref{Dimensionality Partition}). Second, we construct BB-trees in the partitioned subspaces and integrate them to form a BB-forest (described in Section~\ref{BB-Forest}). Third, we transform the data points into tuples for computing the searching bound (described in Section~\ref{Derivation of Bound}). During the search processing, we first transform the query point into a triple and compute the bound used for the range query (described in Section~\ref{Derivation of Bound}). Second, we perform the range query for the candidates. Finally, the \emph{k}NN results are evaluated from these candidates. The whole process is illustrated in Fig~\ref{Visio-Process}.
\begin{figure}\centering
\includegraphics[width=0.48\textwidth]{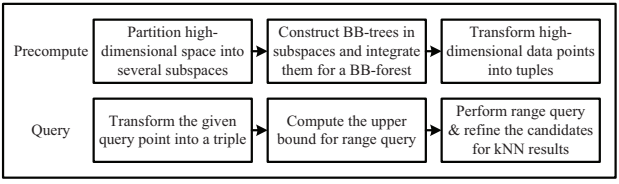}
\caption{Overview} \label{Visio-Process}
\end{figure}
We summarize the frequently-used symbols in Table~\ref{Notations}.
\begin{table}\centering
\caption{Frequently used symbols}
\label{Notations}
\begin{tabular}{ll}
\hline\noalign{\smallskip}
Symbol & Explanation \\
\noalign{\smallskip}\hline\noalign{\smallskip}
$S$ & dataset \\
$n$ & number of data points \\
$d$ & dimensionality of each data point \\
$k$ & number of returned points that users require\\
$x$ & data point \\
$y$ & query point \\
$P(x)$ & transformed data point \\
$Q(y)$ & transformed query point \\
$D_{f}(x,y)$ & the Bregman distance between $x$ and $y$ \\
$UB(x,y)$ & the upper bound of Bregman distance between $x$ and $y$ \\
$M$ & number of partitions \\
\noalign{\smallskip}\hline
\end{tabular}
\end{table}

\section{Derivation of Bound}\label{Derivation of Bound}
For \emph{k}NN queries, it is crucial to avoid the exhaustive search by deriving an effective bound as a pruning condition. We exploit the property of Bregman distances and an upper bound is derived from Cauchy inequality.

Given a data set $D$, suppose ${x}=(x_{1},...,x_{d})^{\rm T}$ and ${y}=(y_{1},...,y_{d})^{\rm T}$ are two $d$-dimensional vectors in $D$. After dimensionality partitioning, $x$ is partitioned into $M$ disjoint subspaces represented by $M$ subvectors. These $M$ subvectors are denoted by:
\begin{align*}
{x}_{i\cdot}&={\left({x}_{\left\lceil\frac{d}{M}\right\rceil\times\left(i-1\right)+1},...,{x}_{\left\lceil\frac{d}{M}\right\rceil\times i}\right)}^{\rm T}\\
&=(x_{i1},...,x_{i\left\lceil\frac{d}{M}\right\rceil})^{\rm T},
\end{align*}
where $1\le i\le M$. Vector $y$ is partitioned in the same manner, while each part is denoted by ${y}_{i\cdot}$ $\left(1\le i\le M\right)$. By Cauchy inequality, we can prove Theorem~\ref{Theorem 1} below, which can be used to derive the upper bound between arbitrary $x_{i\cdot}$ and $y_{i\cdot}$ ($1\le i\le M$) in the same subspace.

\newtheorem{theorem1}{Theorem}
\begin{theorem1}\label{Theorem 1}
The upper bound between $x_{i\cdot}$ and $y_{i\cdot}$ ($1\le i\le M$) is derived:
$${D}_{f}\left(x_{i\cdot},y_{i\cdot}\right)\le {\alpha}_{x}^{(i)}+{\alpha}_{y}^{(i)}+\beta_{yy}^{(i)}+\sqrt{\gamma_{x}^{(i)}\times \delta_{y}^{(i)}}.$$
For simplicity, ${\alpha}_{x}^{(i)}$, ${\alpha}_{y}^{(i)}$, $\beta_{yy}^{(i)}$, $\gamma_{x}^{(i)}$ and $\delta_{y}^{(i)}$ are used to mark these formulas:
$${\alpha}_{x}^{(i)}=\sum_{j=1}^{\left\lceil d/M \right\rceil}{f\left({x}_{ij}\right)}, {\alpha}_{y}^{(i)}=-\sum_{j=1}^{\left\lceil d/M \right\rceil}{f\left({y}_{ij}\right)},$$
$$\beta_{xy}^{(i)}=-\sum_{j=1}^{\left\lceil d/M \right\rceil}({x}_{ij}\times \frac{\partial f(y)}{\partial y_{ij}}), \beta_{yy}^{(i)}=\sum_{j=1}^{\left\lceil d/M \right\rceil}(y_{ij}\times \frac{\partial f(y)}{\partial y_{ij}}),$$
$$\gamma_{x}^{(i)}=\sum_{j=1}^{\left\lceil d/M \right\rceil} x_{ij}^{2}, \delta_{y}^{(i)}=\sum_{j=1}^{\left\lceil d/M \right\rceil}{\frac{\partial f(y)}{\partial y_{ij}}}^{2}.$$
\end{theorem1}
\begin{proof}
Please see Section 1 in the supplementary file.
\end{proof}

Based on Theorem~\ref{Theorem 1}, we can derive the upper bound between two arbitrary data points in each subspace. Furthermore, we can obtain the upper bound between two data points in the original high-dimensional space by summing up the $M$ upper bounds from all the subspaces. Theorem~\ref{Theorem 2} proves it.

\begin{theorem1}\label{Theorem 2}
The Bregman distance between $x$ and $y$ is bounded by the sum of each upper bound from its corresponding subspace. Formally, $D_{f}({x},{y})\le {{UB}({x},{y})}$, where ${{UB}({x},{y})}=\sum_{i=1}^{M}{{UB}\left({x}_{i\cdot},{y}_{i\cdot}\right)}$.
\end{theorem1}
\begin{proof}
Please see Section 2 in the supplementary file.
\end{proof}

Therefore, when dealing with a \emph{k}NN search, we first compute the upper bounds between the given query points and every data point in the dataset. And then the $k^{th}$ smallest upper bound is selected and its components are selected as the searching bounds to perform range queries in their corresponding subspaces as the filter processing. The union of the searching result sets of all the subspaces is the final candidate set for the refinement processing. Theorem~\ref{Theorem 3} proves that all the \emph{k}NN points are in the final candidate set.

\begin{theorem1}\label{Theorem 3}
The candidate set of each subspace is denoted by $C_{i}$ ($1\leq i\leq M$), and the final candidate set:
$$C=C_{1}\cup C_{2}\cup\cdot\cdot\cdot\cup C_{M}.$$
The $k$NN points of the given query point $y$ exist in $C$.
\end{theorem1}
\begin{proof}
Please see Section 3 in the supplementary file.
\end{proof}

Based on the upper bound derived in Theorem~\ref{Theorem 1}, two components ${\alpha}_{x}$ and $\gamma_{x}$ of data points in all the subspaces can be computed offline. It can be considered as the precomputation that each partitioned multi-dimensional data point in each subspaces is transformed into a two-dimensional tuple denoted by $P(x)=({\alpha}_{x},\gamma_{x})$ (See Fig.~\ref{Visio-Transformation} for an illustration). Once the query point is given, we only need to compute the values of ${\alpha}_{y}$, $\beta_{yy}$ and $\delta_{y}$, which can be considered as a triple denoted by $Q(y)=({\alpha}_{y},\beta_{yy},\delta_{y})$, to obtain the upper bounds while the computation overhead is very small. Relying on the precomputation, the search processing will be dramatically accelerated.

\begin{figure}\centering
\includegraphics[width=0.3\textwidth]{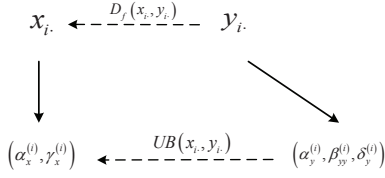}
\caption{Transformation} \label{Visio-Transformation}
\end{figure}

Algorithm~\ref{UBCompute} describes the process of computing the upper bound between two arbitrary data points by their corresponding transformed vectors. Algorithm~\ref{PTransform} and Algorithm~\ref{QTransform} describe the transformations of the partitioned data point and the partitioned query point, respectively. Algorithm~\ref{QBDetermine} describes the process of determining the searching bound from each subspace.

\begin{algorithm}[htb]
  \small\caption{\textsf{UBCompute} ($P(x)$,$Q(y)$)}
  \label{UBCompute}
  \begin{algorithmic}[1]
    \Require
    the transformed vectors of $x$ and $y$, $P(x)=({\alpha}_{x},\gamma_{x})$ and $Q(y)=({\alpha}_{y},\beta_{yy},\delta_{y})$.
    \Ensure
    the upper bound of $x$ and $y$, $ub$.
    \State $ub={\alpha}_{x}+{\alpha}_{y}+\beta_{yy}+\sqrt{\gamma_{x}\times \delta_{y}}$;\\
    \Return $ub$;
    \end{algorithmic}
\end{algorithm}

\begin{algorithm}[htb]
  \small\caption{\textsf{PTransform} ($X$)}
  \label{PTransform}
  \begin{algorithmic}[1]
    \Require
    the subvectors' set of a partitioned data point $X=\{x_{1},x_{2},...,x_{M}\}$.
    \Ensure
    the transformed tuples' set of a partitioned data point $P=\{P(x_{1}),P(x_{2}),...,P(x_{M})\}$.
      \For{$i=1$ to $M$}
        \State Compute $\alpha_{x}^{(i)}$ and $\gamma_{x}^{(i)}$;
        \State $P(x_{1})=(\alpha_{x}^{(i)},\gamma_{x}^{(i)})$;
      \EndFor
      \State $P=\{P(x_{1}),P(x_{2}),...,P(x_{M})\}$;\\
      \Return $P$;
    \end{algorithmic}
\end{algorithm}

\begin{algorithm}[htb]
  \small\caption{\textsf{QTransform} ($Y$)}
  \label{QTransform}
  \begin{algorithmic}[1]
    \Require
    the subvectors' set of the partitioned query point $Y=\{y_{1},y_{2},...,y_{M}\}$.
    \Ensure
    the transformed tuples' set of the partitioned query point $Q=\{Q(y_{1}),Q(y_{2}),...,Q(y_{M})\}$.
      \For{$i=1$ to $M$}
        \State Compute $\alpha_{y}^{(i)}$, $\beta_{yy}^{(i)}$ and $\delta_{y}^{(i)}$;
        \State $Q(x_{i})=(\alpha_{y}^{(i)},\beta_{yy}^{(i)},\delta_{y}^{(i)})$;
      \EndFor
      \State $Q=\{Q(y_{1}),Q(y_{2}),...,Q(y_{M})\}$;\\
      \Return $Q$;
    \end{algorithmic}
\end{algorithm}

\begin{algorithm}[htb]
  \small\caption{\textsf{QBDetermine} ($S_{t}$,$Q$)}
  \label{QBDetermine}
  \begin{algorithmic}[1]
    \Require
    the transformed dataset $S_{t}=\{P_{1},P_{2},...,P_{n}\}$, the transformed query point $Q$.
    \Ensure
    the set containing $M$ subspaces' searching bounds $QB$.
      \State $UB=\emptyset$, $QB=\emptyset$;
      \For{$i=1$ to $n$}
        \For{$j=1$ to $M$}
          \State $ub_{ij}=\textsf{UBCompute} (P_{ij},Q_{j})$; // Algorithm~\ref{UBCompute}
          \State $ub_{i}+=ub_{ij}$;
          \State $QB_{i}=QB_{i}\cup \{ub_{ij}\}$;
      \EndFor
      \State $UB=UB\cup {ub_{i}}$;
    \EndFor
    \State Sort $UB$ and find the $k^{th}$ smallest $ub_{t}$ ($1\leq t\leq n$);
    \State $QB=QB_{t}=\{ub_{t1},ub_{t2},...,ub_{tM}\}$.\\
    \Return $QB$;
    \end{algorithmic}
\end{algorithm}
\vspace{-15pt}
\section{Dimensionality Partitioning}\label{Dimensionality Partition}
According to Cauchy inequality, we can prove:
$$\sum_{i=1}^{M_{1}}{\sqrt{\gamma_{x}^{(i)}\times \delta_{y}^{(i)}}}>\sum_{i=1}^{M_{2}}{\sqrt{\gamma_{x}^{(i)}\times \delta_{y}^{(i)}}},$$
when $M_{1}<M_{2}$. It indicates that more partitions bring tighter bounds. Furthermore, we can prove that there is an exponential relationship between the upper bound and the number of partitions. However, more partitions lead to more online computation overhead. Therefore, it's an unresolved issue to determine the number of partitions. Besides, how to partition these dimensions may also affect the efficiency. In this section, the optimized number of partitions which contributes to the efficiency is derived theoretically and a heuristic strategy is devised for the purpose of reducing the candidate set by partitioning.
\subsection{Number of Partitions}\label{Number of Partitions}
More partitions lead to a tighter bound, which indicates a smaller candidate set intuitively. But more partitions also incur more computation overhead when calculating the upper bounds from more subspaces at the same time. It's a trade-off between the number of partitions and the efficiency. In this part, we will theoretically analyze the online processing to derive the optimized number of partitions for the optimized efficiency. The online processing mainly includes two parts, one is to compute the upper bounds between the given query point and an arbitrary data point in the dataset and sort these upper bounds to determine the searching bound. The other is to obtain the candidate points by performing range queries and refine the candidate points for the results. Below we will derive the time complexities of the online \emph{k}NN search.

As a prerequisite, we transform the given query into a triple:
$$Q(y)=(-\sum f(y),\sum{\frac{\partial f(y)}{\partial y}},\sum{y\times\frac{\partial f(y)}{\partial y}}).$$
When it comes to computing upper bounds, these triples in the $M$ partitions can be computed in $O(d)$ time. Since we have already transformed the multi-dimensional data points in each subspace into tuples, we can compute the upper bounds in $O(Mn)$ time in all $M$ subspaces. The time complexity of summing up all $n$ points' upper bounds is also $O(Mn)$ and the time complexity of sorting them to find the $k^{th}$ smallest one is $O(n\log k)$. Therefore, the whole time complexity of transforming the query point, computing the upper bounds, summing up them and sorting them to find the $k^{th}$ smallest one is $O(d+Mn+n\log k)$.

In the filter-refinement processing, we specify a parameter $\lambda$ $(0<\lambda <1)$ to describe the pruning effect of the searching bound for the complexity analyses. We suppose that data points are stored in BB-trees and each leaf node is full. And we set the capacity of each leaf node in a BB-tree as $C$ and the number of leaf nodes is $n/C$. Thus the number of accessed leaf nodes during the range queries can be estimated as $\lambda n/C$. According to the proposed method in~\cite{DBLP:conf/nips/Cayton09}, the secant method is employed to determine whether a cluster and the searching range intersect, or whether one contains another when performing the range queries. For the time consumption of the determination is negligible, the time complexity of searching in the $M$ BB-trees is equivalent to that of searching in a binary tree and the time complexity is $O(M\times\frac{\lambda n}{C}\log{\frac{n}{C}})$. In BB-trees, the capacity of each leaf node increases with the size of the corresponding dataset since we have to restrict the height of the tree, and the value of $\frac{n}{C}$ can be considered as a constant. Consequently, the time complexity of searching in the BB-trees is ignored.

Once we have identified the clusters that intersect the given searching range, all the data points in these clusters will be loaded from the disk into the memory for processing as the candidate points which will be refined to obtain the \emph{k}NN points. Thus we should estimate the size of the candidate set primarily. There is an exponential relationship between the upper bound $UB$ and the number of partitions $M$ and we describe it as $UB=A\alpha^{M}$, where $0<\alpha<1$. In addition, we assume the parameter $\lambda$ which describes the pruning effect is proportional to the upper bound $UB$ and is represented as $\lambda=\beta UB$. Based on these, the number of the candidate points is $\beta A\alpha^{M}n$. Actually, the final candidate set is the union of all the subspaces' candidate subsets. Here we directly consider the candidate set obtained by searching in the original space with the summing searching bound as the final candidate set, since it is a very good approximation of the union of all the subspaces' candidate subsets which is verified by experiments. Therefore, the refinement process takes $O(\beta A\alpha^{M}nd+\beta A\alpha^{M}n\log k)$. The first item represents the time complexity of computing the Bregman distances between each candidate point and the query, while the second represents the time complexity of evaluating the \emph{k}NN points.

Therefore, the online time complexity is
$$O(d+Mn+n\log k+\beta A\alpha^{M}nd+\beta A\alpha^{M}n\log k)$$
in total and we target at minimizing the total time cost.

\begin{theorem1}\label{Theorem 4}
For any user-specified $k$ $(0<k\le n)$, the time complexity can be minimized by setting the number of partitions $$M=\log_{\alpha}{\frac{2n}{-\mu \ln{\alpha}(d+\log k)}},$$
where $\mu=\beta An$.
\end{theorem1}
\begin{proof}
Please see Section 4 in the supplementary file.
\end{proof}
Especially, since the number of partitions requires to be determined offline, we set the value of $k$ to $1$ which will not impact too much on the number of partitions since $k$ is negligible compared to the value of $n$. Moreover, the result computed by Theorem~\ref{Theorem 4} may not be an integer, hence we compute the time costs in both cases of rounding up and down and choose the best value of $M$. $A$ and $\alpha$ can be determined by fitting the function $UB=A\alpha^{M}$ through two arbitrary points' $UB$ and the corresponding $M$, while the points are randomly selected from the dataset. And $\beta$ can be determined by computing the proportion of the points within each sample's $UB$ to $n$.

As $M$ increases, the I/O cost decreases exponentially, and the corresponding time consumption can be estimated as $\frac{\beta A\alpha^{M}n}{Bv}$, where $B$ denotes the disk's page size and $v$ denotes the disk's IOPS. Since the IOPS of current mainstream SSD is very high, the loss of the I/O's time consumption can be negligible compared to the gains in the CPU's running time for the optimized partition number. When using the low-level storage device, this time consumption incurred by I/O operations can be added to the above cost model for deriving the optimized number of partitions.

\subsection{Partitioning Strategy}\label{Partition Strategy}
Initially, we simply choose an equal and contiguous strategy for the dimensionality partitioning. Since the final candidate set is the union of the candidate subsets of all the subspaces, the size of the candidate set depends on the size of the candidate subset of each partition. Therefore, we attempt to develop a strategy for the dimensionality partitioning to further reduce the size of the final candidate set.

Intuitively, when the number of the candidate points generated from each partition is constant, if the intersection between these partitions is small, the candidate set will be large. In the worst case, the final candidate set reaches the maximum in the case that their intersection is empty. Conversely, large portions of overlap among these partitions' candidate sets may lead to a smaller candidate set. At best, each of them is equivalent to the final candidate set. Therefore, our objective is to reduce the size of the entire candidate set by making the partitions' candidate set overlap as large and as possible.

Simply, if there exist two identical partitions, their corresponding candidate sets will overlap completely leading to a smaller candidate set. Therefore, it is crucial to measure and compare the similarities between different partitions. However, there is no reasonable indicator satisfying our requirements. To address this issue, we simplify partitions to dimensions and employ Pearson Correlation Coefficient to measure the correlations between different dimensions and the correlations are utilized to indicate the similarities. The dimensions with strong correlations will be assigned to different partitions to make dimensions in each partition uniformly distributed. It is a heuristic algorithm and named Pearson Correlation Coefficient-based Partition (PCCP) for the dimensionality partitioning.

Given a \emph{d}-dimensional dataset, we attempt to divide \emph{d} dimensions into \emph{M} partitions while each partition contains $\left\lceil d/N\right\rceil$ dimensions. The correlation between two arbitrary dimensions $X$ and $Y$ is measured by the Pearson correlation coefficient $r(X,Y)=cov(X,Y)/\sqrt{var(X)var(Y)}$, where $cov(X,Y)$ represents the covariance between $X$ and $Y$, and $var(X)$ and $var(Y)$ represent the variances of $X$ and $Y$, respectively. Specifically, we only consider the level of dimensions' correlations and ignore whether it is positive or negative, so we only consider the absolute values of the Pearson correlation coefficients. Our proposed strategy consists of two steps:

\textbf{Assignment}: We assign \emph{d} dimensions to \emph{M} groups and attempt to ensure high similarities between the dimensions in each group. In detail, we first select a dimension randomly and insert it into a group, find the dimension having the largest Pearson correlation coefficient with the selected dimension and insert it into the same group. Then, the dimension which has the largest absolute Pearson correlation coefficient with an arbitrary inserted dimension is selected and inserted into the current group. Following these steps, we continue searching and inserting until \emph{M} dimensions have been inserted and a group is formed. The above steps are repeated until all \emph{d} dimensions are assigned to $\left\lceil d/M\right\rceil$ different groups.

\textbf{Partitioning}: We select a dimension from every group and insert it into the current partition so that each partition has $\left\lceil d/N\right\rceil$ dimensions. The above steps are repeated until the \emph{M} partitions are obtained.

Fig.~\ref{Visio-PCCP} illustrates the process by an example. There exists a six dimensional dataset whose dimensions are denoted by $a$, $b$, $c$, $d$, $e$ and $f$. In the assignment, we randomly select a dimension such as $a$, and compute the correlations between $a$ and other dimensions. We find $|r(a,e)|$ is the largest, so we assign $a$ and $e$ to the first group. Whereafter, we compute the correlations between $e$ and other dimensions since we have computed the correlations between $a$ and other dimensions. We find $|r(e,f)|$ is the largest among the correlations between $a$ and other dimensions except $e$, as well as the correlations between $e$ and other dimensions except $a$. So we insert $f$ into the group containing $a$ and $e$. The others, $b$, $c$ and $d$, are assigned to the second group. After assignment, we randomly select a dimension from each of the two groups and insert them into a partition. The above procedure is repeated three times and the final partitions are obtained as Fig.~\ref{Visio-PCCP}.
\begin{figure}
\centerline{\includegraphics[width=0.36\textwidth]{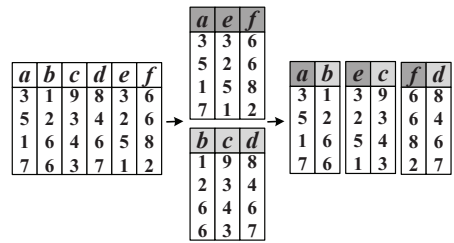}}
\caption{An example of PCCP} \label{Visio-PCCP}
\end{figure}
\vspace{-10pt}
\section{Indexing Subspaces for \emph{k}NN Search}\label{BB-Forest}
After partitioning, the original high-dimensional space is divided into many low-dimensional subspaces. This enables us to use existing indexing techniques which are efficient for low-dimensional spaces such as~\cite{DBLP:conf/icml/Cayton08, DBLP:journals/pvldb/ZhangOPT09}. Although the previous BB-tree~\cite{DBLP:conf/icml/Cayton08} is only designed for the \emph{k}NN search, an algorithm is explored to specifically solve the range query based on BB-tree in~\cite{DBLP:conf/nips/Cayton09}, which shows good performance. Therefore, we employ BB-tree for the range queries in each partition's filtering process and adopt the algorithm in~\cite{DBLP:conf/nips/Cayton09} in this paper. We construct a BB-tree for each subspace and all the BB-trees form a \emph{BB-forest}. Note that directly conducting \emph{k}NN queries in each partition cannot obtain a correct candidate. Thus the \emph{k}NN algorithm and the filtering condition proposed in~\cite{DBLP:conf/icml/Cayton08} cannot be simply extended using the partition strategy. On the other hand, the range query algorithm proposed in~\cite{DBLP:conf/nips/Cayton09} cannot directly solve the \emph{k}NN problem. Therefore, our framework is essentially different from the existing BB-tree based filtering algorithms.

Intuitively, if we store each node's corresponding cluster's center and radius, and their pointers to their children or the addresses of the data points in this cluster, BB-tree can be simply extended to the disks to process large-scale datasets. Fig.~\ref{BB-forest} illustrates the disk-resident integrated index structure which consists of $M$ BB-trees constructed in the $M$ partitioned subspaces. In each BB-tree, the intermediate nodes store their corresponding clusters' centers and radii, denoted by $C.center$ and $C.radius$, respectively. And the leaf nodes store not only clusters' centers and radii, but the addresses of the data points in their corresponding clusters, denoted by $P.address$, which are used for reading the original high-dimensional data points from the disks.

According to~\cite{DBLP:conf/nips/Cayton09}, which develops an algorithm for efficient range search with Bregman distances, only each cluster's center and radius are required to determine whether two clusters intersect or one cluster contains another during the search processing. Therefore, the disk-resident version of BB-tree still works for efficient range search with Bregman distances. However, since our proposed BB-forest is an integrated index structure, how to organize the data on the disk is a major issue to address. If the clusters in these subspaces are significantly different, different BB-trees' candidate nodes obtained by ranges queries will index different data points which causes more I/O accesses when reading high-dimensional data points from the disks.

Benefitting from our devised PCCP, similar clusters are obtained in different subspaces. Therefore, after the dimensionality partitioning, we construct a BB-tree in a randomly selected subspace and the original high-dimensional data points indexed by each leaf node are organized according to the order of these leaf nodes. At the same time, the address of each point containing the disk number and offset is recorded. During the construction of other BB-trees, the recorded addresses of these data points are stored in leaf nodes for indexing. Since the candidate points in different subspaces are the same at best, we will read the same part of the disks when performing range queries in different subspaces to reduce I/O accesses.
\begin{figure}
\centerline{\includegraphics[width=0.35\textwidth]{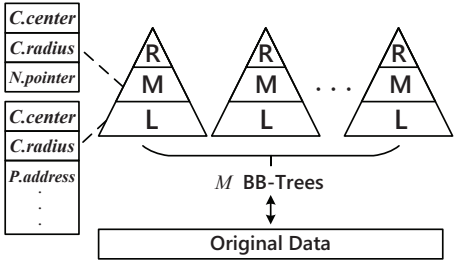}}
\smallskip
\caption{BB-forest} \label{BB-forest}
\end{figure}
\vspace{-10pt}
\section{Overall Framework}\label{Putting Everything Together}
In this paper, we solve the \emph{k}NN search problem with Bregman distances in high-dimensional space by our proposed partition-filter-refinement framework. The framework consists of the precomputation in Algorithm~\ref{BrePartition1} and the search processing in Algorithm~\ref{BrePartition2}.
In Algorithm~\ref{BrePartition1}, we first determine the number of partitions (Line 2), which is introduced in Section~\ref{Number of Partitions}. Then, we perform the dimensionality partitioning in the dataset based on our proposed strategy PCCP denoted as $\textsf{PCCP} (S)$ (Line 3), which is introduced in Section~\ref{Partition Strategy}. Finally, we transform the original data points in each subspace into tuples used for computing the upper bounds later (Lines 4-7), which is introduced in Section~\ref{Derivation of Bound}. At the same time, the partitioned data points will be indexed by the BB-trees of all the subspaces forming the BB-forest (Line 8), and the construct of the BB-forest is denoted as $\textsf{BBFConstruct} (S_{p})$ which is introduced in Section~\ref{BB-Forest}.
\begin{algorithm}[htb]
  \small\caption{\textsf{BrePartitionConstruct} ($S$,$n$,$d$,$A$,$\alpha$,$\beta$)}
  \label{BrePartition1}
  \begin{algorithmic}[1]
    \Require
    $S$, $n$, $d$, $A$, $\alpha$, $\beta$.
    \Ensure
    $S_{t}$, $BBF$.
      \State $S_{t}=\emptyset$, $BBForest=\emptyset$;
      \State $M=\log_{\alpha}{\frac{2n}{-\beta And\ln{\alpha}}}$;
      \State $S_{p}=\textsf{PCCP} (S)$; 
      \For{each $X\in S_{p}$}
        \State $P=\textsf{PTransform} (X)$; // Algorithm~\ref{PTransform}
        \State $S_{t}=S_{t}\cup P$;
      \EndFor
      \State $BBF=\textsf{BBFConstruct} (S_{p})$;
      \Return $S_{t}$, $BBF$;
    \end{algorithmic}
\end{algorithm}

In Algorithm~\ref{BrePartition2}, we first rearrange the query point according to the partitioned data points (Line 2). Second, we transform the query point into triples in each subspace (Line 3) and compute the searching bounds from all the subspaces (Line 4), which is introduced in Section~\ref{Derivation of Bound}. Third, we perform range queries with the computed searching bounds over all the BB-trees in the BB-forest for retrieving the candidates (Lines 5-7). Finally, we evaluate the candidate points for the \emph{k}NN points and return the result set (Lines 8-9).
\begin{algorithm}[htb]
  \small\caption{\textsf{BrePartitionSearch} ($S_{p}$,$y$,$k$,$S_{t}$,$BBF$)}
  \label{BrePartition2}
  \begin{algorithmic}[1]
    \Require
    $S_{p}$, $y$, $k$, $S_{t}$, $BBF$.
    \Ensure
    $Res$.
      \State $Cand=\emptyset$, $Res=\emptyset$;
      \State Rearrange the query point to obtain $Y$;
      \State $Q=\textsf{QTransform} (Y)$; // Algorithm~\ref{QTransform}
      \State $QB=\textsf{QBDetermine} (S_{t},Q)$; // Algorithm~\ref{QBDetermine}
      \For{each $T_{j}\in BBF$}
        \State $Cand=Cand\cup T_{j}.\textsf{rangeQuery} (Y[j],QB[j])$;
      \EndFor
      \State Evaluate the \emph{k}NN result set $Res=\{x_{1},x_{2},...,x_{k}\}$;\\
      \Return $Res$;
    \end{algorithmic}
\end{algorithm}
\vspace{-10pt}
\section{Extension to Approximate \textit{k}NN Search}\label{Extension}
In previous sections, we mainly concentrate on the exact retrieval of \textit{k}NN searches. However, through our research, we can tighten the bounds derived above for comparable approximate results with probability guarantees to improve the efficiency. We present an approximate solution with probability guarantees. By the solution, given a query point $q$ and a probability guarantee $p$ ($0<p\leq 1$), the retrieval $k$ points are the exact \emph{k}NN points of $q$ with the probability guarantee $p$.

According to the previous description in Theorem~\ref{Theorem 1}, we relax the item $\beta_{xy}$ by employing Cauchy-inequality and an exact bound for the \emph{k}NN search can be derived. The exact searching bound can be simply represented in the form of $\mathcal{\kappa+\mu}$, where
$$\mathcal{\kappa}=\sum_{i=1}^{d}{f\left({x}_{i}\right)}-\sum_{i=1}^{d}{f\left({y}_{i}\right)}+\sum_{i=1}^{d}(y_{i}\times \frac{\partial f(y)}{\partial y_{i}})$$
which isn't affected when computing the upper bound, and
$$\mathcal{\mu}=\sqrt{\sum_{i=1}^{d} x_{i}^{2}\times \sum_{i=1}^{d}{\frac{\partial f(y)}{\partial y_{i}}}^{2}}$$
which is obtained by relaxing $\beta_{xy}$. Therefore, the searching bound is actually determined by $\mathcal{\mu}$. In our proposed approximate \emph{k}NN search solution, $\mu$ is tightened by multiplying a coefficient denoted as $c$ ($0<c\leq 1$) in the context that the distribution of each dimension's data is known. With this tighter bound, we can obtain the \emph{k}NN results more efficiently within a smaller bound with the probability guarantee. Therefore, we mainly focus on how to derive the value of $c$ when the value of $\mathcal{\mu}$ is known.
\newtheorem{proposition}{Proposition}
\begin{proposition}\label{App_Prop}
We set the following two events with respect to $A$ and $B$, which describe the exact condition and the approximate condition, respectively. And $C$ indicates that Bregman distances are non-negative.
$$A:\beta_{xy}<\mathcal{\mu}, B:\beta_{xy}\leq c\mathcal{\mu} (0<c\leq 1), C:\beta_{xy}\geq -\mathcal{\kappa}.$$
When the distribution of each dimension in the dataset is known\footnote{Histograms can be used to describe each dimension's distribution and a known distribution similar to the histogram (such as Normal distribution) can be chosen to fit the dimension's distribution by the least squares method.} and the given probability guarantee is $p$, the value of $c$ is:
$$c=\Psi_{\beta_{xy}}^{-1}(p\Psi_{\beta_{xy}}(\mathcal{\mu})+(1-p)\Psi_{\beta_{xy}}(-\mathcal{\kappa}))/\mathcal{\mu},$$
where $\Psi_{\beta_{xy}}$ and $\Psi_{\beta_{xy}}^{-1}$ are used to denote the variable $\beta_{xy}$'s cumulative distribution function (cdf) and its inverse function, respectively.
\end{proposition}
\begin{proof}
Please see Section 5 in the supplementary file.
\end{proof}
\begin{figure}
\centerline{\includegraphics[width=0.4\textwidth]{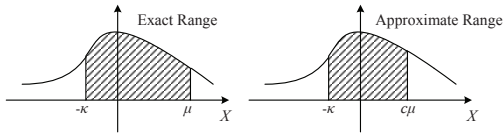}}
\caption{Exact and Approximate Cases}\label{Visio-distribution}
\end{figure}

Proposition~\ref{App_Prop} illustrates that the exact bound is tightened by multiplying an approximate coefficient to obtain the approximate results with probability guarantees in the condition that the distribution of each dimension in the dataset is known. Fig.~\ref{Visio-distribution} shows the exact and approximate cases intuitively. When given the exact bound, we firstly compute the approximate coefficient in the original space and the exact bounds obtained in all partitions are multiplied by the approximate coefficient to be each partition's new approximate bound.
\vspace{-10pt}
\section{Experiments}\label{Experiments}
We conduct extensive experiments on four real datasets and two synthetic datasets to verify the efficiency of our methods. In this section, we present the experimental results.
\vspace{-10pt}
\subsection{Experiment Setup}
\subsubsection{Benchmark Method}
According to TABLE~\ref{Summary}, we select two state-of-the-art techniques~\cite{DBLP:journals/pvldb/ZhangOPT09, DBLP:conf/icml/Cayton08}, which are designed for the exact \emph{k}NN search for Bregman distances. We denote them by "VAF" and "BBT", respectively. Our method is denoted by "BP" (the abbreviation of "BrePartition") and its approximate versions are denoted by "ABP" (the abbreviation of "Approximate BrePartition"). Three key parameters are shown in TABLE~\ref{Parameters}. We measure the index construction time, I/O cost and CPU time to verify the efficiency of the benchmarks and our proposed method. These methods are implemented in Java. All experiments were done on a PC with Intel Core i7-2600M 3.40GHz CPU, 8 GB memory and 1 TB WD Blue 3D NAND SATA SSD, running Windows x64.
\begin{table}[H]\centering
\caption{Parameters}
\label{Parameters}
\begin{tabular}{ll}
\hline\noalign{\smallskip}
Parameter & Varying Range \\
\noalign{\smallskip}\hline\noalign{\smallskip}
Result Size $k$ & 20, 40, 60, 80, 100 \\
Dimensionality (Fonts) & 100, 200, 300, 400 \\
Data Size (Sift) & 2M, 4M, 6M, 8M, 10M \\
\noalign{\smallskip}\hline
\end{tabular}
\end{table}
\subsubsection{Datasets}
Four real-life datasets Audio\footnote{http://www.cs.princeton.edu/cass/audio.tar.gz}, Fonts\footnote{http://archive.ics.uci.edu/ml/datasets/Character+Font+Images}, Deep\footnote{http://yadi.sk/d/I$\_$yaFVqchJmoc}, Sift\footnote{https://archive.ics.uci.edu/ml/datasets/SIFT10M} and two synthetic datasets are summarized in TABLE~\ref{Dataset}. Normal is a 200-dimensional synthetic dataset which has 50000 points generated by simulating the random number of standard normal distribution. Uniform is a 200-dimensional synthetic dataset which has 50000 points generated by simulating the random number of uniform distribution between $[0,100]$. Normal and Uniform are only used for evaluating our proposed approximate solution. For all these datasets, 50 points are randomly selected as the query sets. And we randomly select 50 samples for computing $A$, $\alpha$ and $\beta$.
\begin{table}[H]
\centering
\caption{Datasets}
\label{Dataset}
\begin{tabular}{llllll}
\hline\noalign{\smallskip}
Parameter & $n$ & $d$ & $M$ & Page Size & Measure\\
\noalign{\smallskip}\hline\noalign{\smallskip}
Audio & 54387 & 192 & 28 & 32KB & ED \\
Fonts & 745000 & 400 & 50 & 128KB & ISD \\
Deep & 1000000 & 256 & 37 & 64KB & ED \\
Sift & 11164866 & 128 & 22 & 64KB & ED \\
Normal & 50000 & 200 & 25 & 32KB & ED \\
Uniform & 50000 & 200 & 21 & 32KB & ISD \\
\noalign{\smallskip}\hline
\end{tabular}
\end{table}

\begin{figure}[H]\centering
\includegraphics[width=0.35\textwidth]{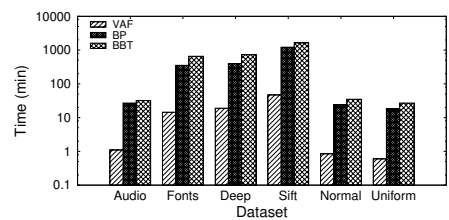}
\caption{Index Construction Time} \label{Index Construction Time}
\end{figure}

\begin{figure*}
\hspace{-10pt}
\subfigure[Audio]{
\includegraphics[width=.25\textwidth]{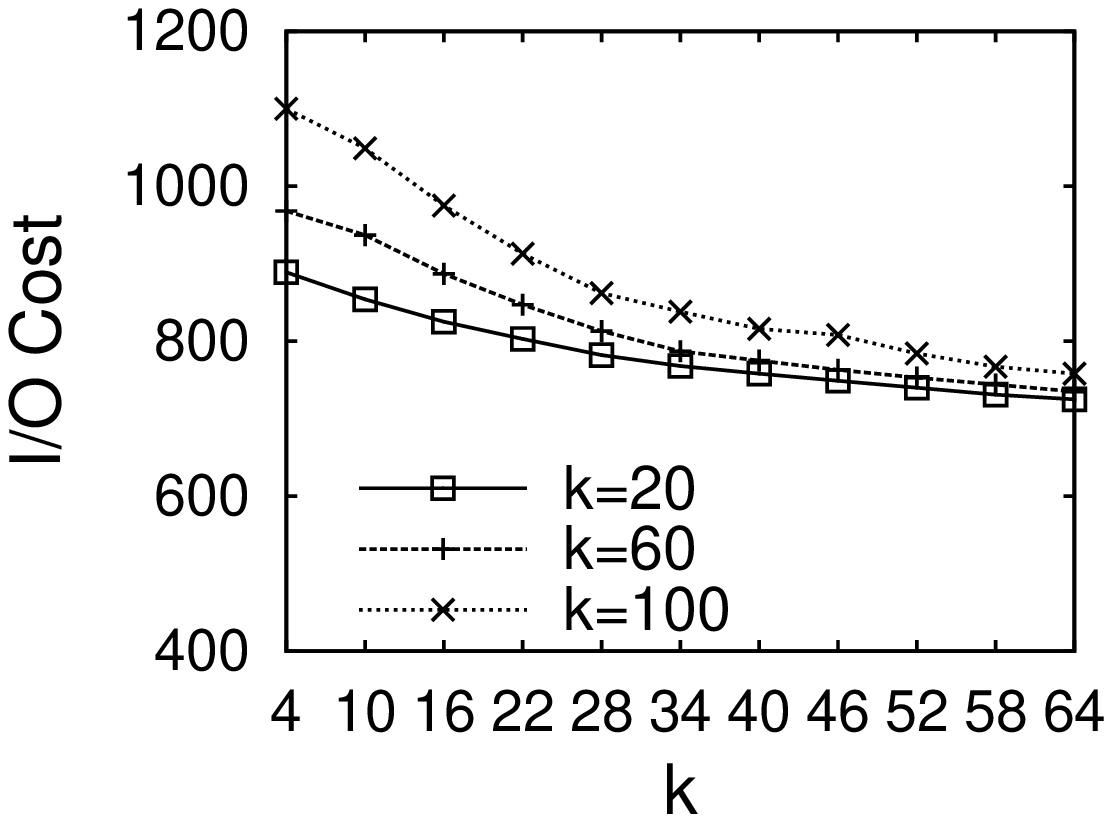}}
\hspace{-10pt}
\subfigure[Fonts]{
\includegraphics[width=.25\textwidth]{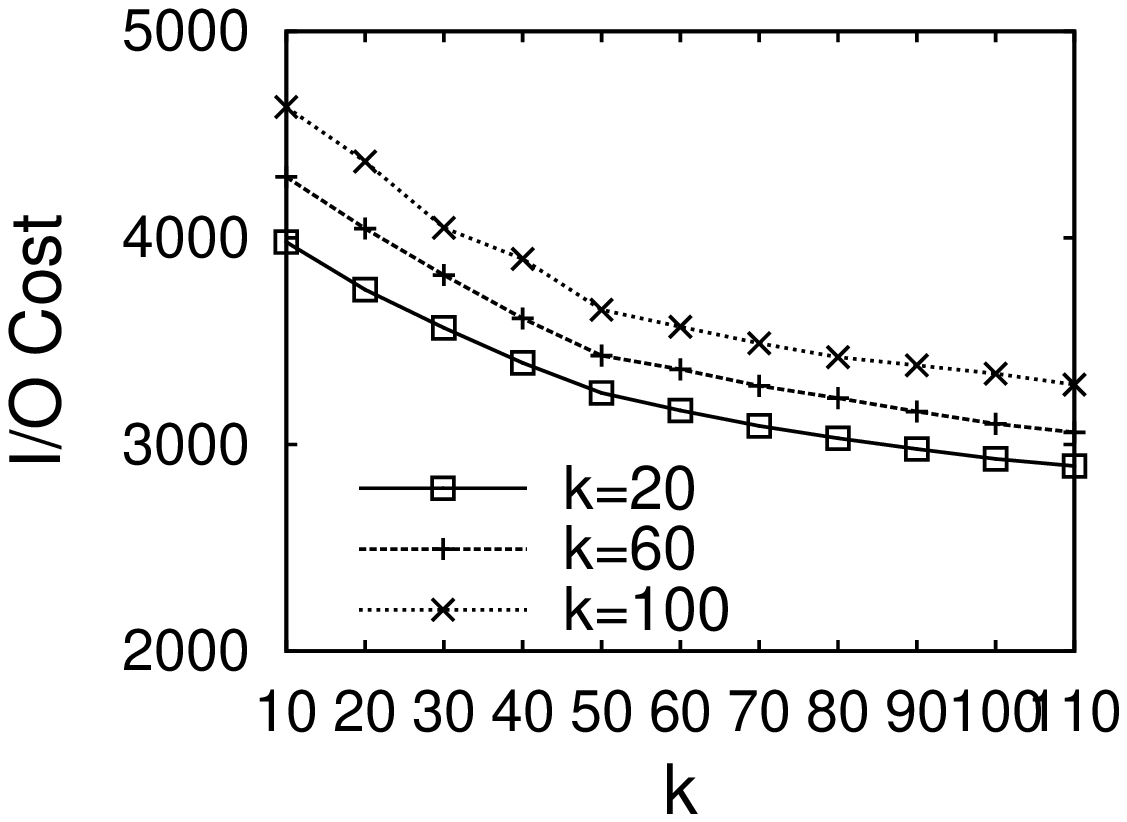}}
\hspace{-10pt}
\subfigure[Deep]{
\includegraphics[width=.25\textwidth]{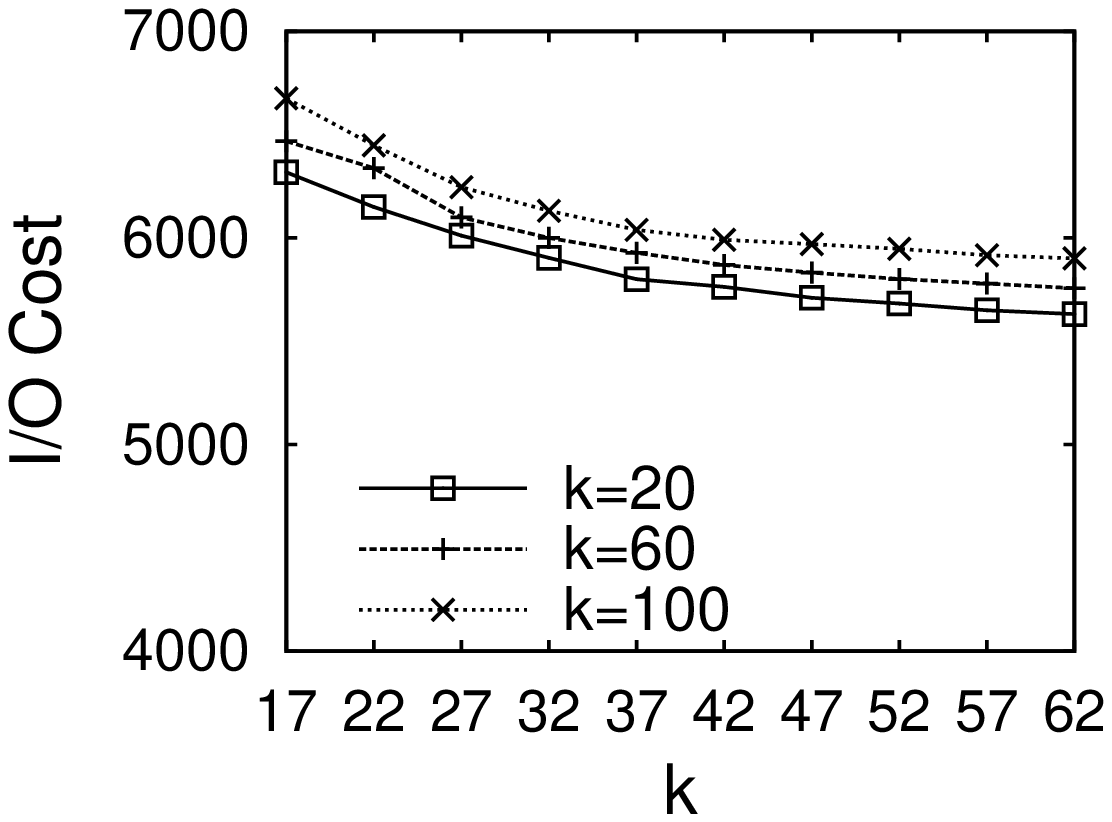}}
\hspace{-10pt}
\subfigure[Sift]{
\includegraphics[width=.25\textwidth]{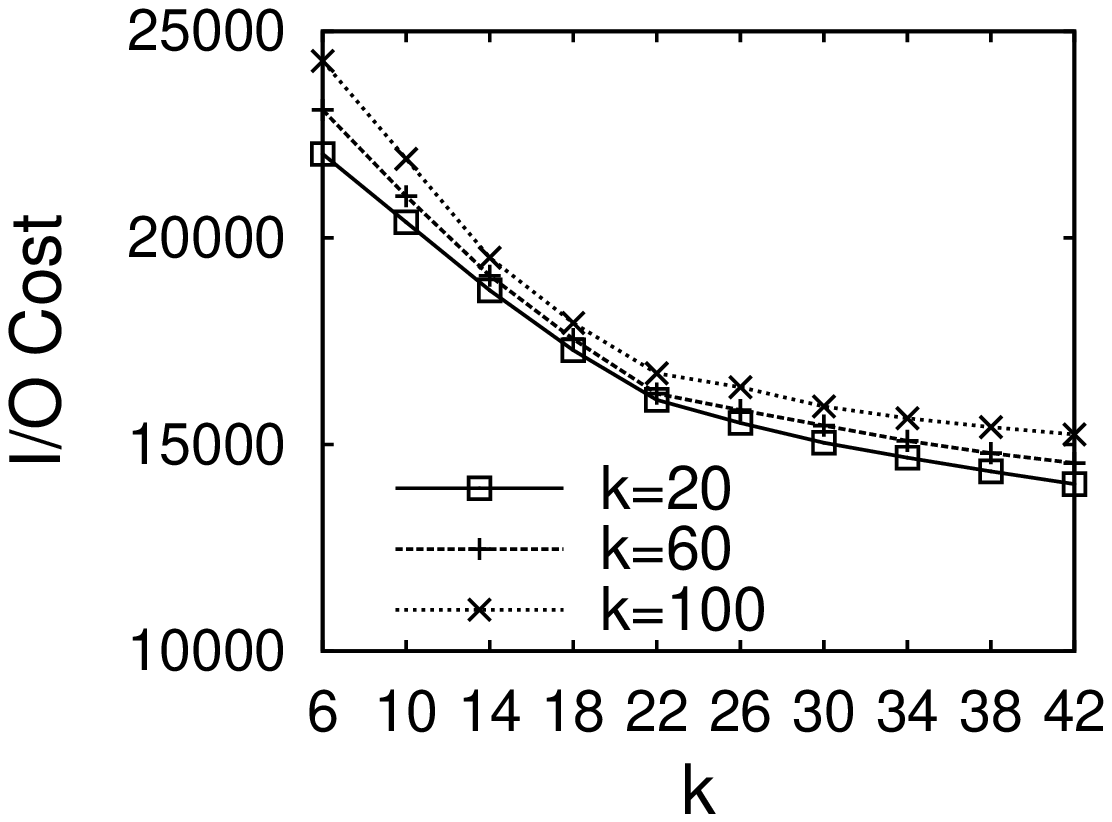}}
\caption{I/O Cost (Impact of $M$)}
\label{I/O Cost (Impact of $M$)}
\end{figure*}

\begin{figure*}
\hspace{-10pt}
\subfigure[Audio]{
\includegraphics[width=.25\textwidth]{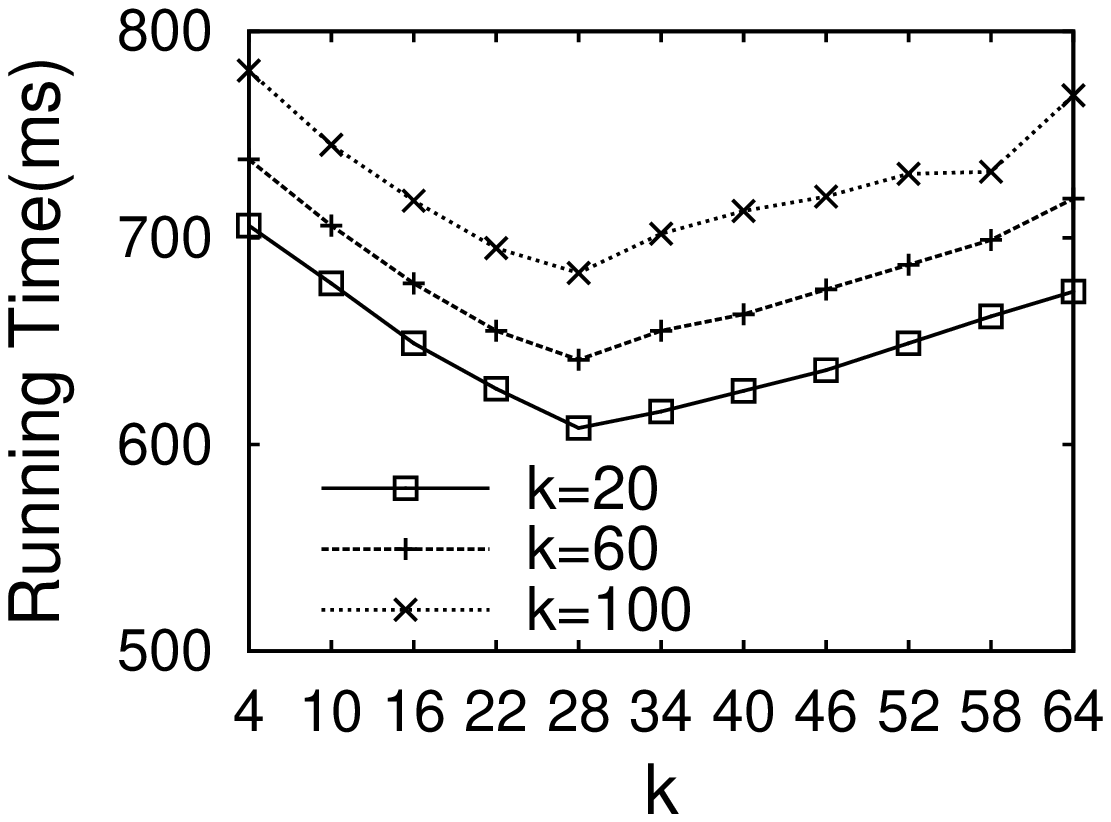}}
\hspace{-10pt}
\subfigure[Fonts]{
\includegraphics[width=.25\textwidth]{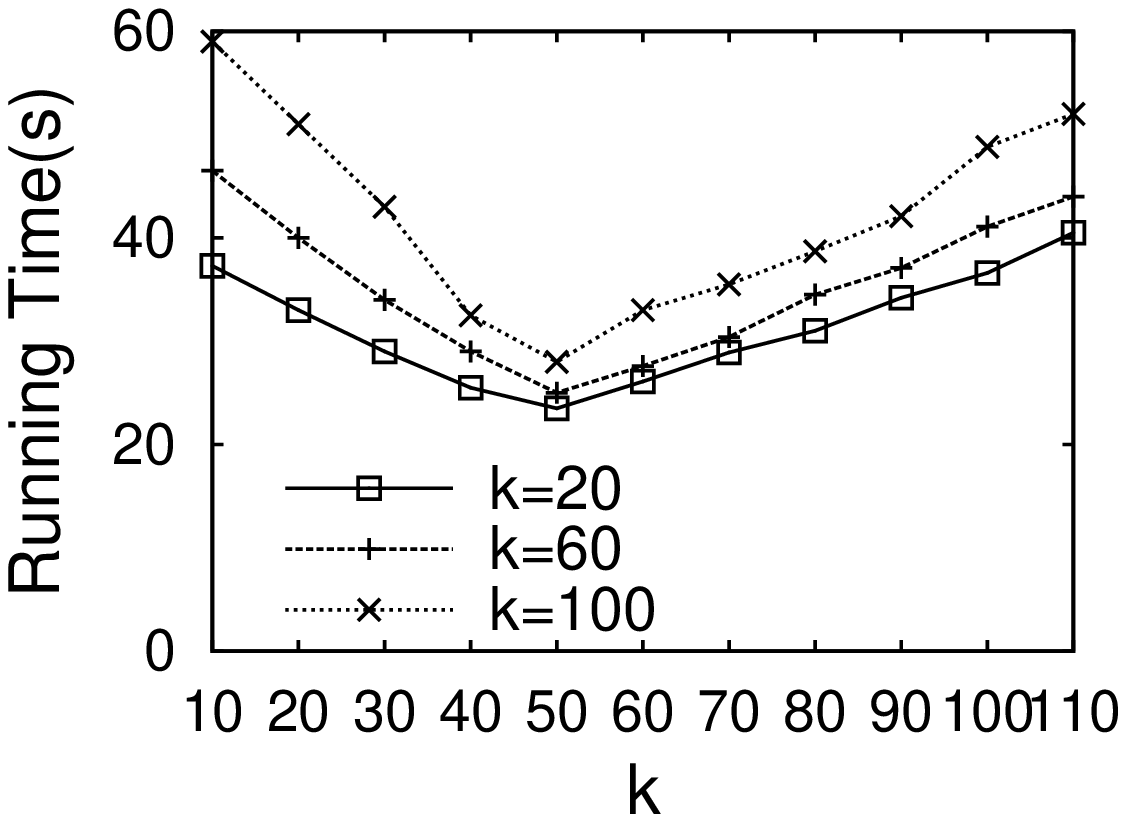}}
\hspace{-10pt}
\subfigure[Deep]{
\includegraphics[width=.25\textwidth]{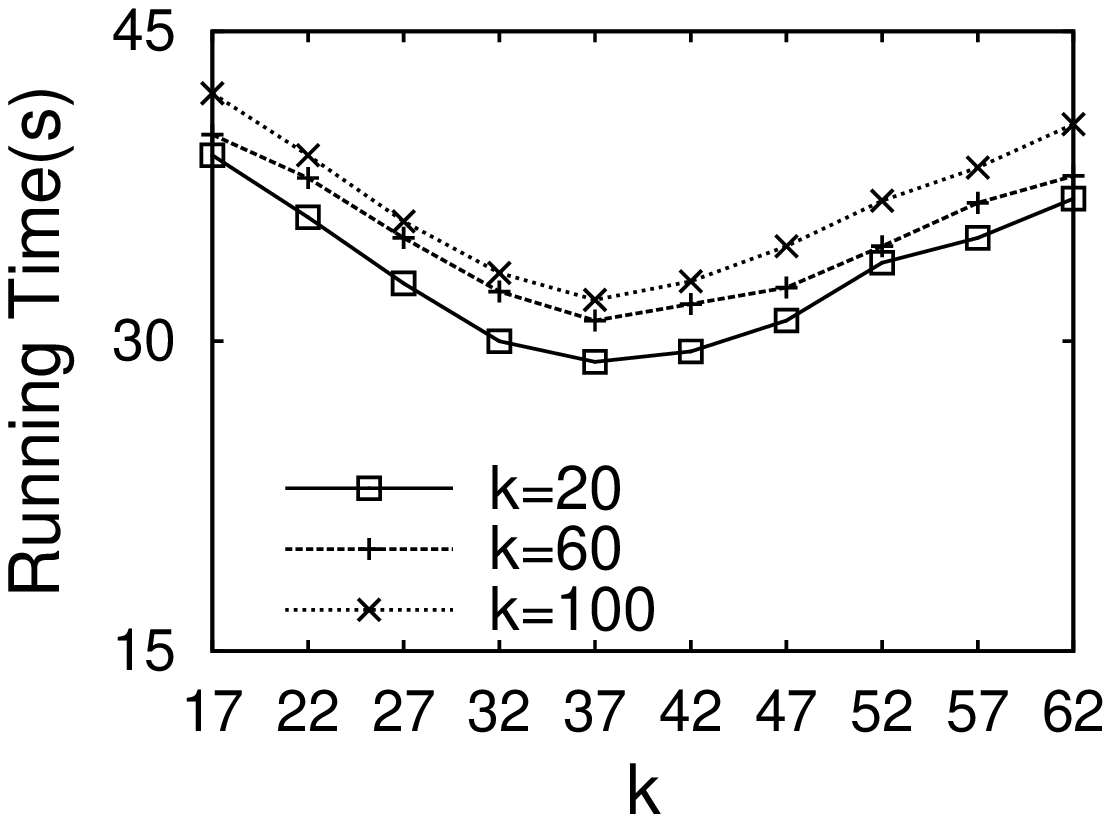}}
\hspace{-10pt}
\subfigure[Sift]{
\includegraphics[width=.25\textwidth]{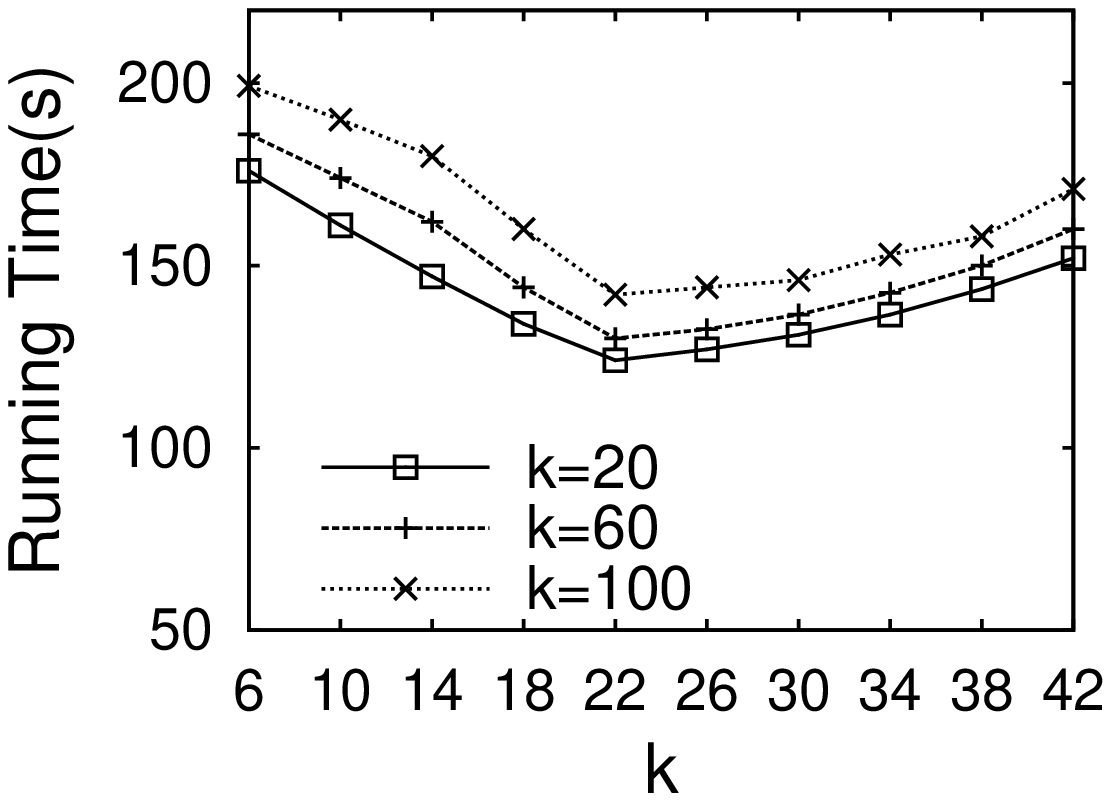}}
\caption{Running Time (Impact of $M$)}
\label{Running Time (Impact of $M$)}
\end{figure*}
\vspace{-10pt}
\subsection{Index Construction Time}
In Fig.~\ref{Index Construction Time}, we illustrate the index construction time of all the three testing methods on six testing datasets. The construction of VA-file is the fastest among all the methods on all the datasets, while the index structures constructed based on Bregman balls, such as BB-forest and BB-tree, increase the construction time by at least one magnitude because of the time-consuming clustering process. And the construction of BB-tree is slower than our proposed BB-forest since the clustering becomes less effective with the increase of dimensionality. In addition, it's indicated that the construction of VA-file becomes slower in high-dimensional spaces as well.
\vspace{-10pt}
\subsection{Impact of Dimensionality Partitioning}
In our method, we develop the dimensionality partitioning technique to solve the \emph{k}NN search problem in the high-dimensional space. Through theoretical analysis, the optimized number of partitions $M$ can be derived and we obtain the optimized numbers of partitions on the four datasets according to Theorem~\ref{Theorem 4}, which are shown in TABLE~\ref{Dataset}. In addition, a novel strategy named PCCP are explored. In this part, we will evaluate the impact of the parameter $M$ and PCCP on the efficiency of our method and validate that the derived value of $M$ is optimized.
\subsubsection{Impact of the number of partitions $M$}
We show the results on four real datasets to demonstrate the impact of $M$ on the efficiency. We vary $k$ to 20, 40 and 60 to evaluate the I/O cost and the running time when the value of $M$ is set within a certain range and the results are shown in Figs.~\ref{I/O Cost (Impact of $M$)} and~\ref{Running Time (Impact of $M$)}. From Fig.~\ref{I/O Cost (Impact of $M$)}, the I/O cost decreases as $M$ increases in all the cases. Moreover, the I/O cost decreases more and more slowly as $M$ increases. In Fig.~\ref{Running Time (Impact of $M$)}, the trend of running time is different from that of the I/O cost, which generally has a trend of descent first and then ascent.
\subsubsection{Validation of the derived optimized value of $M$}
As shown in Fig.~\ref{I/O Cost (Impact of $M$)}, the I/O cost decreases more and more slowly as $M$ increases. This is because the number of candidate points decreases at a slower and slower rate as the searching range decreases at a slower and slower rate. As can be seen in Fig.~\ref{I/O Cost (Impact of $M$)}, the I/O cost is reduced at a low speed after $M$ reaches the optimized value. Meanwhile, from Fig.~\ref{Running Time (Impact of $M$)}, when the values of $M$ are set to our derived optimized number of partitions on the four datasets, the CPU's running time is minimum in all cases which brings up to $30\%$ gains. In addition, the random reads of the current mainstream SSD are about 5k IOPS to 500k IOPS. In this case, the efficiency gains achieved by the number of the I/O reductions can be negligible compared to the gains in the CPU's running time. Therefore, the experimental results validate that our derived value of $M$ is optimized.
\subsubsection{Impact of PCCP}
In this section, we evaluate how PCCP influences the efficiency of our method by testing it in both cases with and without PCCP and the results are shown in Fig.~\ref{Impact of PCCP}. The default value of $k$ is 20. From the experimental results, both the I/O cost and the running time are reduced by $20\%$ to $30\%$ when PCCP is applied, which indicates that PCCP can reduce the number of candidate points leading to lower I/O costs and time consumption. Moreover, the experimental results demonstrate our proposed index structure, BB-forest, can avoid some invalid disk accesses relying on PCCP. We also evaluate the impact of randomly selecting the first dimension in PCCP on the performance. The experimental results are shown in the Section 7 of the supplementary file.
\begin{figure}
\centering
\hspace{-10pt}
\subfigure[I/O cost]{
\includegraphics[width=.23\textwidth]{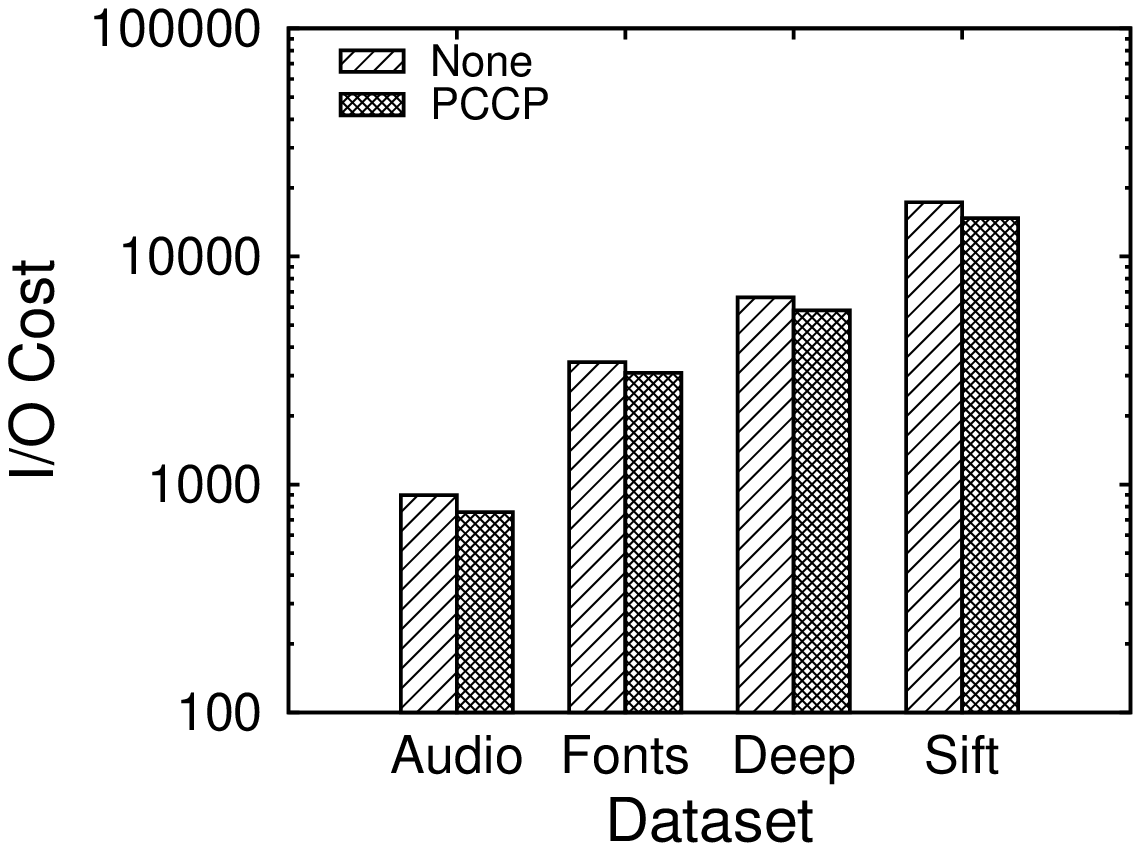}}
\hspace{-10pt}
\subfigure[Running time]{
\includegraphics[width=.23\textwidth]{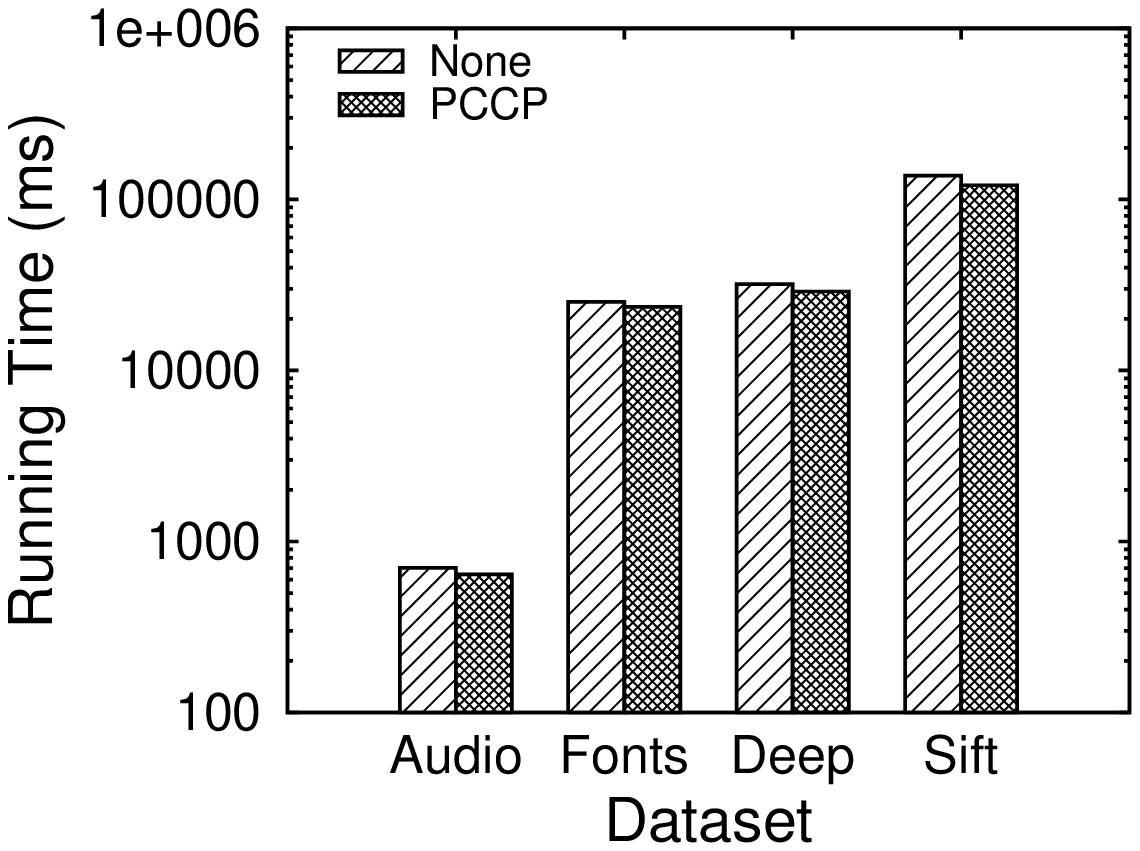}}
\caption{Impact of PCCP}
\label{Impact of PCCP}
\end{figure}
\vspace{-10pt}
\subsection{I/O Cost}\label{Sec I/O Cost}
In this part, we evaluate the I/O cost by varying $k$ from 20 to 100 and show the experimental results in Fig.~\ref{I/O Cost}. We extend the memory-resident BB-tree to a disk-resident index structure following the idea of our proposed BB-forest. As $k$ increases, BrePartition outperforms the other two methods in I/O cost almost in all testing cases. It benefits from the derived bound's good pruning effect, which leads to a small quantity of candidate points. Moreover, BrePartition possesses good performance on all datasets, which shows that our proposed method can be applied to both large-scale and high-dimensional datasets.

From Fig.~\ref{I/O Cost}, we find that the I/O cost of VA-file is lower than that of BB-tree owing to two aspects. On the one hand, the lower and upper bounds computed by their proposed search algorithm are reasonable, even though less tighter than ours. On the other hand, the whole space has been compactly divided using large number of bits via VA-file which can reduce the size of the candidate set and fewer points will be loaded into main memory. Even so, loading more candidate points and all the disk-resident vector approximations causes more extra I/O costs. Moreover, BB-tree performs the worst among all the testing methods. This is because large overlaps among clusters in the high-dimensional space causes more candidates, which should be checked for the \emph{k}NN results.
\vspace{-10pt}
\subsection{Running Time}\label{Sec Running Time}
In Fig.~\ref{Running Time}, we evaluate the CPU's running time by varying $k$ to test the efficiency of our method. In all testing cases, BrePartition performs consistently the best regardless of the dimension and the size of datasets, which shows our method's high scalability. This is mainly because we partition the original high-dimensional space into several low-dimensional subspaces and build our proposed integrated index structure, BB-forest to accelerate the search processing.

Similarly, compared to BrePartition, VA-file doesn't show satisfying results on the CPU's running time. This is because it's required to scan all the vector approximations representing the data points when pruning the useless points for the candidate points. Besides, the reverse mapping from the vector approximations to their corresponding original points when verifying the candidate points is time-consuming as well. Similar to the I/O cost, BB-tree shows the worst results, because too many candidates to be checked cause performance degradation in the high-dimensional space.
\begin{figure*}
\hspace{-10pt}
\subfigure[Audio]{
\includegraphics[width=.25\textwidth]{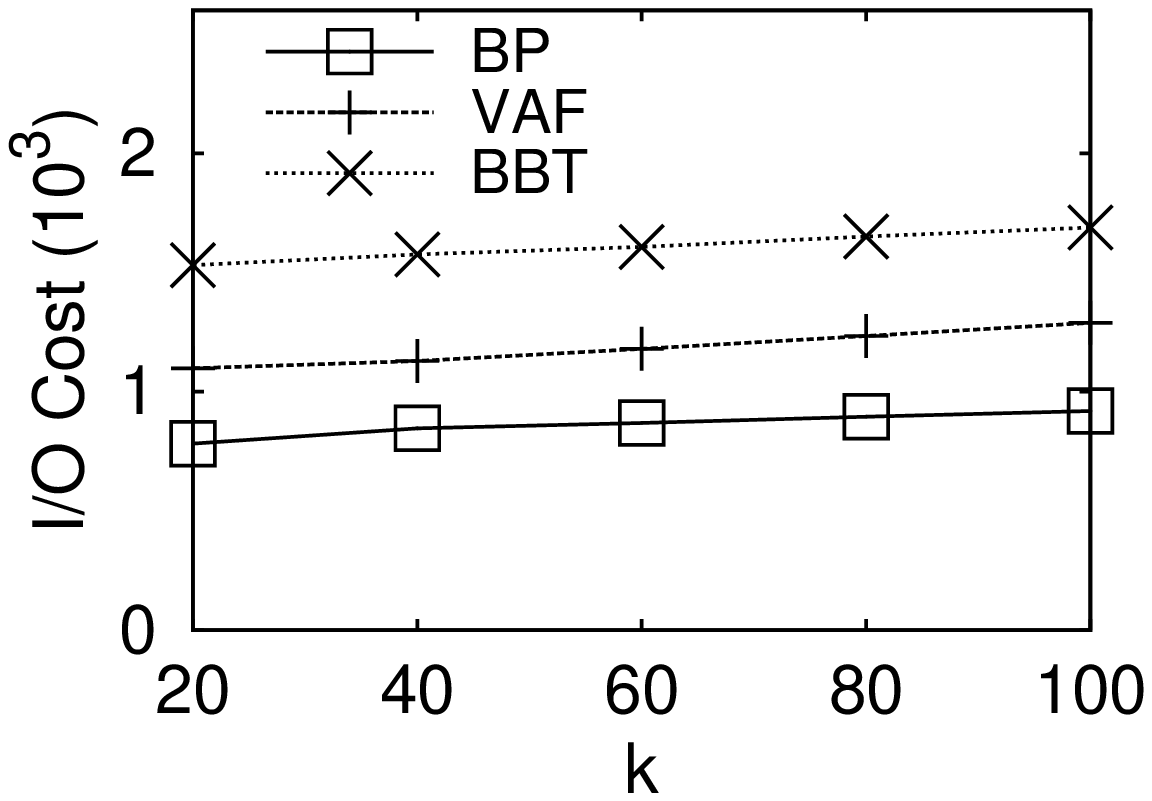}}
\hspace{-10pt}
\subfigure[Fonts]{
\includegraphics[width=.25\textwidth]{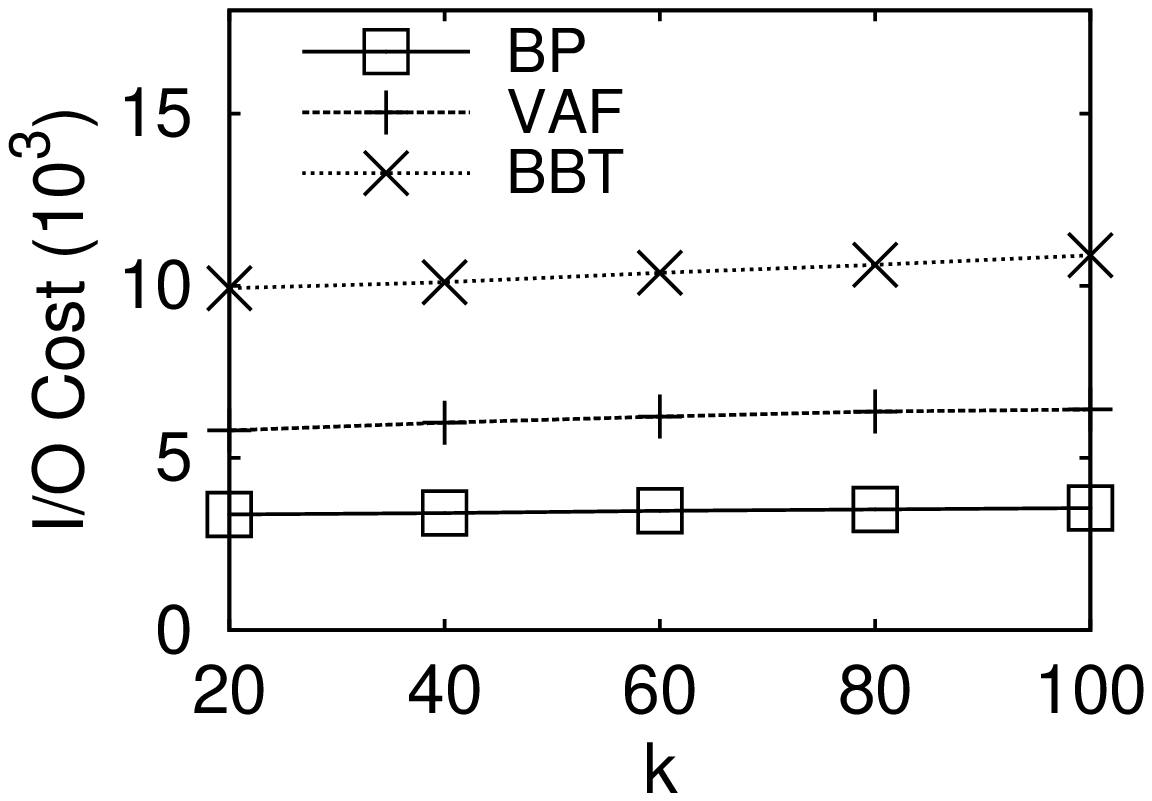}}
\hspace{-10pt}
\subfigure[Deep]{
\includegraphics[width=.25\textwidth]{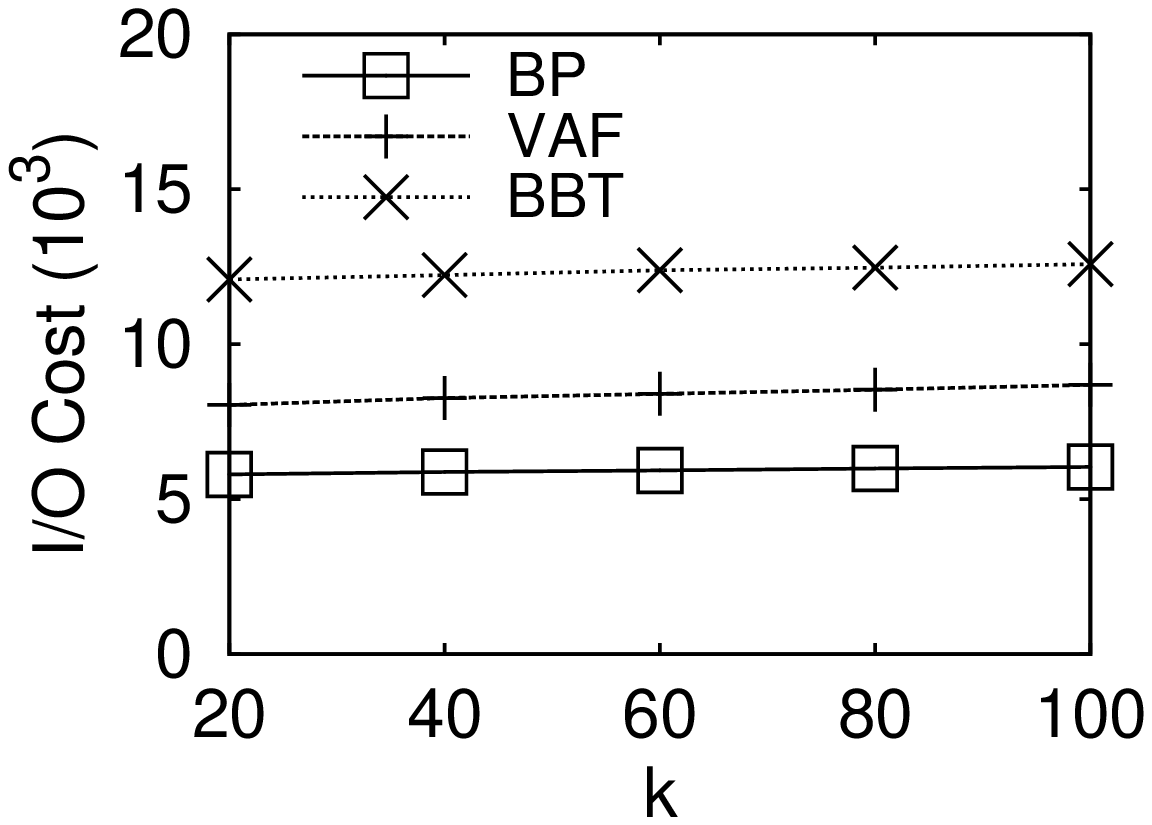}}
\hspace{-10pt}
\subfigure[Sift]{
\includegraphics[width=.25\textwidth]{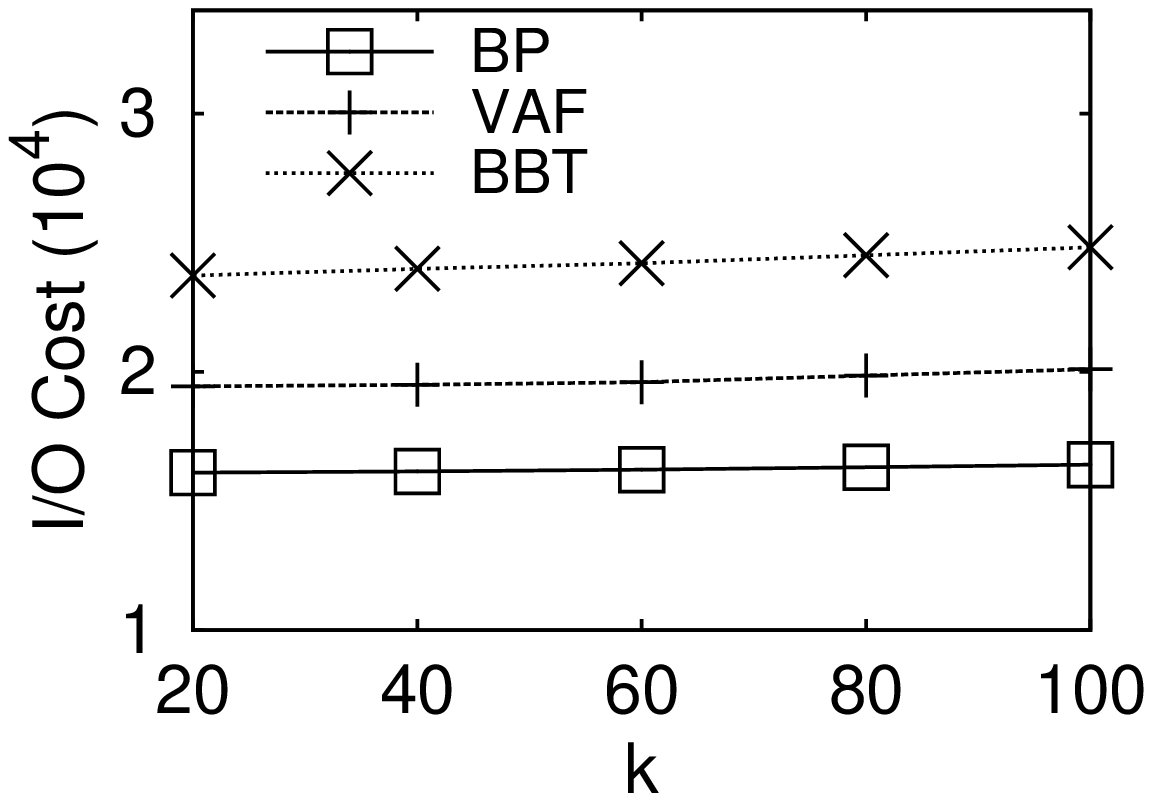}}
\caption{I/O Cost}
\label{I/O Cost}
\end{figure*}

\begin{figure*}
\hspace{-10pt}
\subfigure[Audio]{
\includegraphics[width=.25\textwidth]{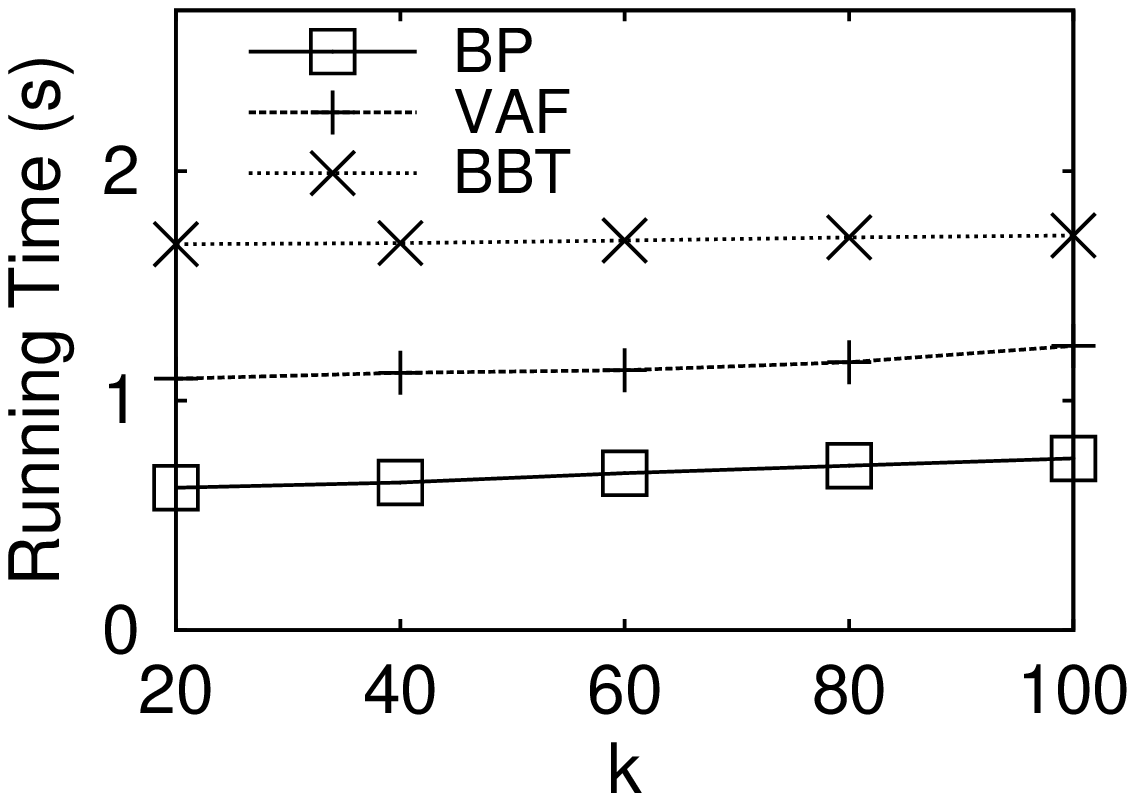}}
\hspace{-10pt}
\subfigure[Fonts]{
\includegraphics[width=.25\textwidth]{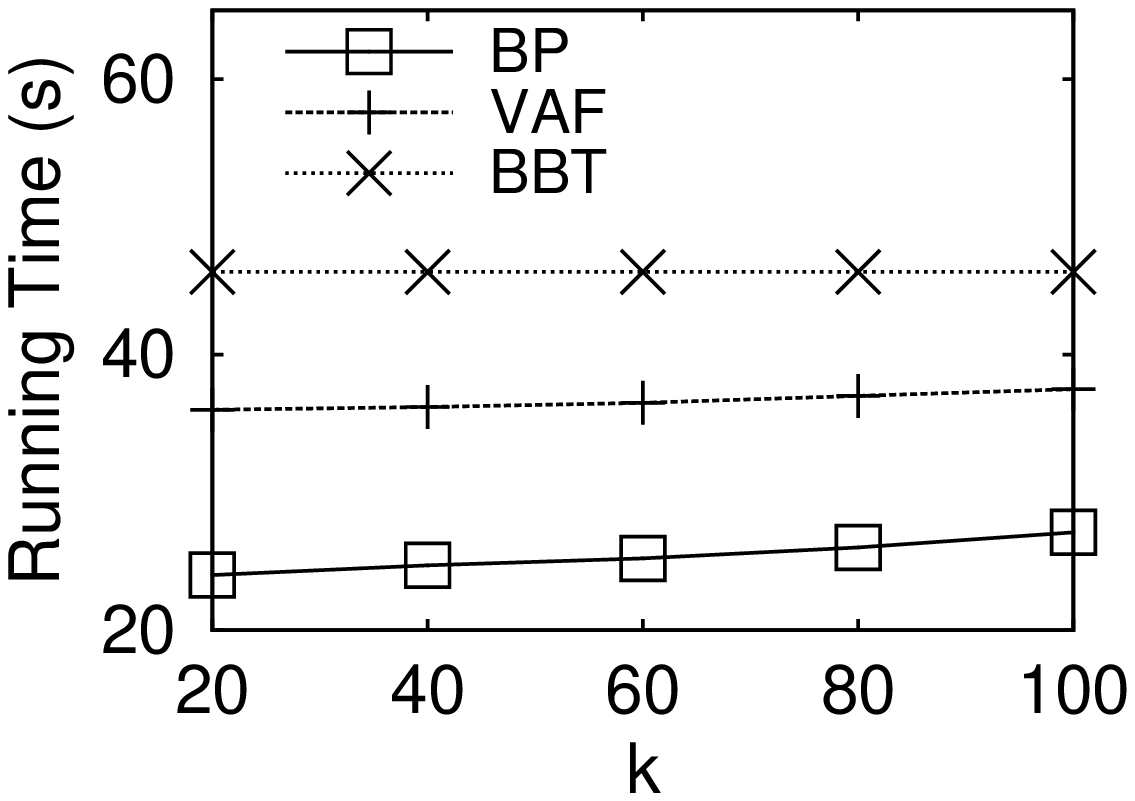}}
\hspace{-10pt}
\subfigure[Deep]{
\includegraphics[width=.25\textwidth]{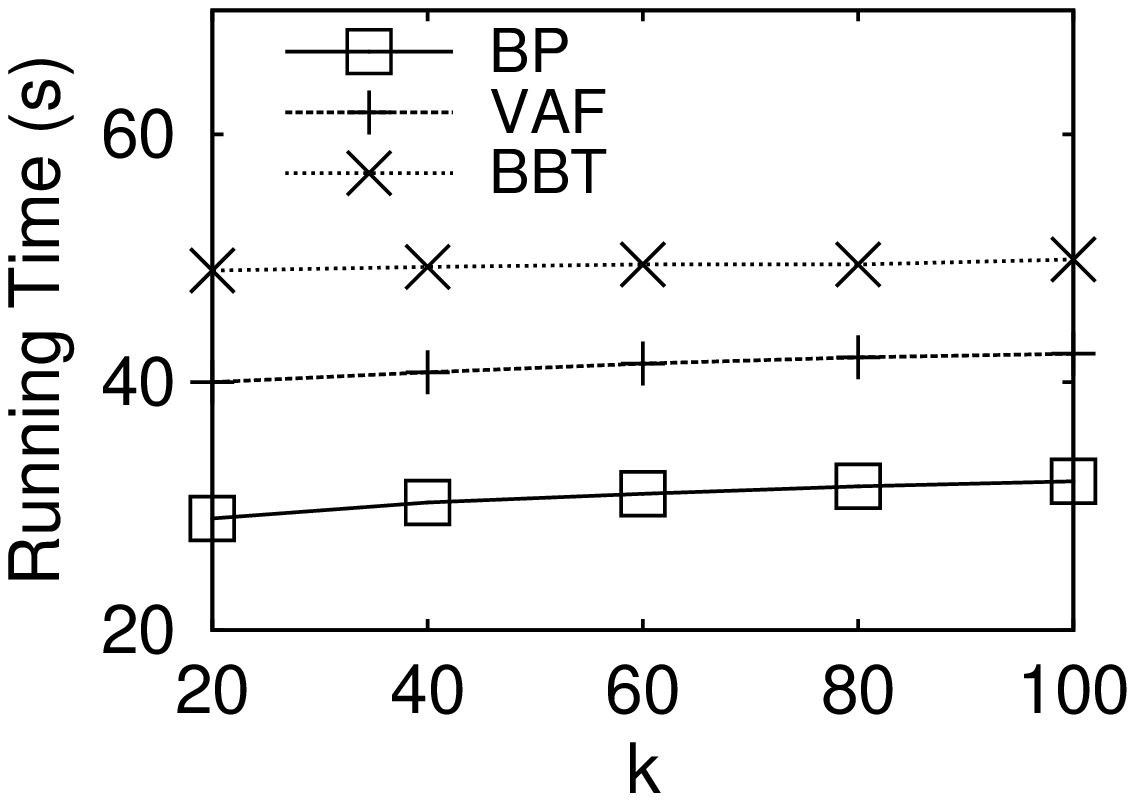}}
\hspace{-10pt}
\subfigure[Sift]{
\includegraphics[width=.25\textwidth]{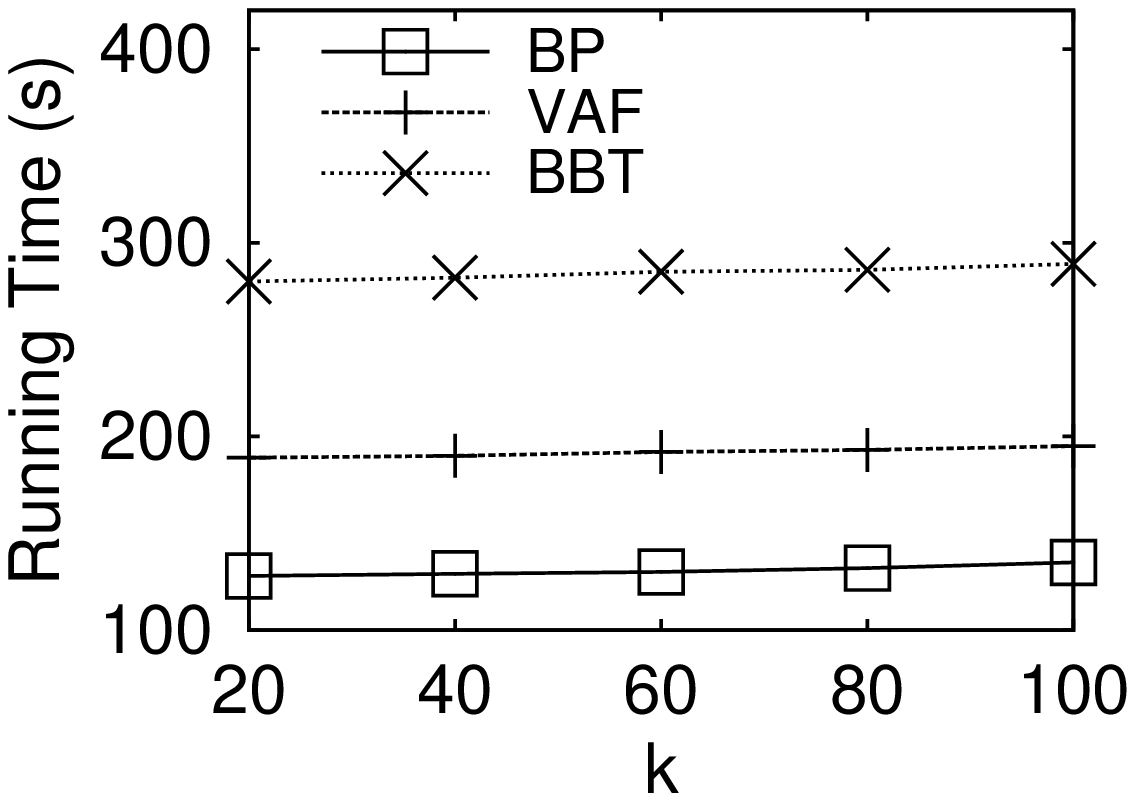}}
\caption{Running Time}
\label{Running Time}
\end{figure*}
\vspace{-10pt}
\subsection{Impact of Dimensionality}
In this section, we evaluate how the dimensionality of a dataset influences the efficiency of these methods. We select Fonts as the testing dataset and evaluate the I/O cost and the CPU's running time in different dimensions of the dataset. We vary the dimensionality from 10 to 400 and the experimental results are shown in Fig.~\ref{Impact of Dimensionality}. In BrePartition, when the dimensionality of testing changes, the number of partitions changes as well. According to Theorem~\ref{Theorem 4}, we set the computed optimized number of partitions to 3, 9, 13, 29 and 50, when the dimensionality is 10, 50, 100, 200 and 400, respectively.

The experimental results illustrate that both the I/O cost and the running time increase monotonically as the dimensionality of the dataset increases. BrePartition demonstrates the best performance as before, and the indicators show slighter increases as the dimensionality increases, which illustrates that it's applicable to various dimensional spaces. This is because BrePartition can derive appropriate bounds for different dimensional spaces. The growth rates of I/O cost and running time in VA-file increase as the dimensionality increases, which shows its poor scalability in high-dimensional space. BB-tree illustrates good performances in 10-dimensional space, but its I/O cost and running time increase significantly, which indicates it's only an efficient method in the low-dimensional space.
\vspace{-10pt}
\subsection{Impact of Data Size}
In this part, we select the dataset Sift and set the data size from 2000000 to 10000000 to test these algorithms' I/O cost and CPU's running time in different scales of datasets. The results are illustrated in Fig.~\ref{Impact of Data Size}. Based on the above Theorem~\ref{Theorem 4}, the value of data size $n$ impacts little on the value of partition's number $M$, so we consistently select 22 as the number of partitions regardless of changes in the data size.

From the experimental results, the I/O cost and the running time almost increase linearly as the size of the dataset increases. BrePartition requires the lowest I/O cost and running time indicating its good scalability with the increasing data amount. In addition, VA-file demonstrates comparable performances in I/O cost and running time as well. Compared to BrePartition, BB-tree's I/O cost and running time are multiplied since it isn't suitable for high-dimensional spaces as explained in Section~\ref{Sec I/O Cost} and~\ref{Sec Running Time}.
\begin{figure}
\hspace{-10pt}
\subfigure[I/O cost]{
\includegraphics[width=.25\textwidth]{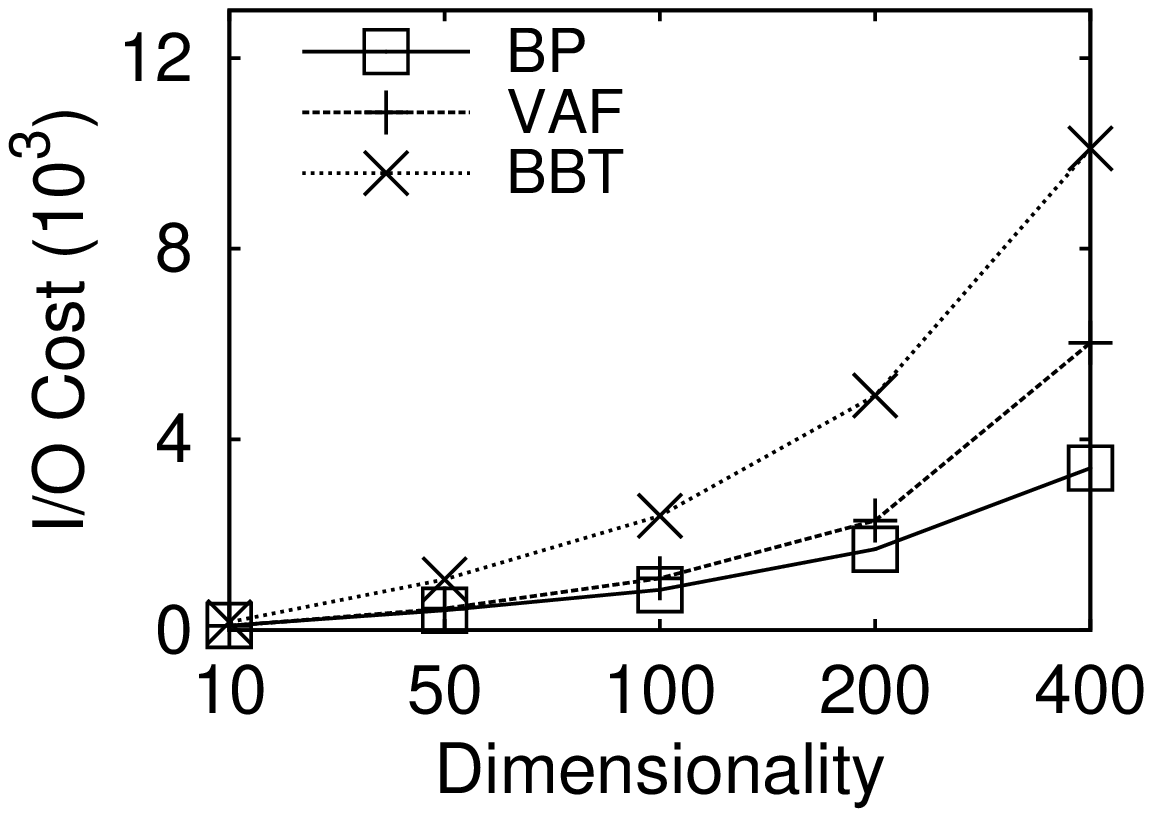}}
\hspace{-10pt}
\subfigure[Running time]{
\includegraphics[width=.25\textwidth]{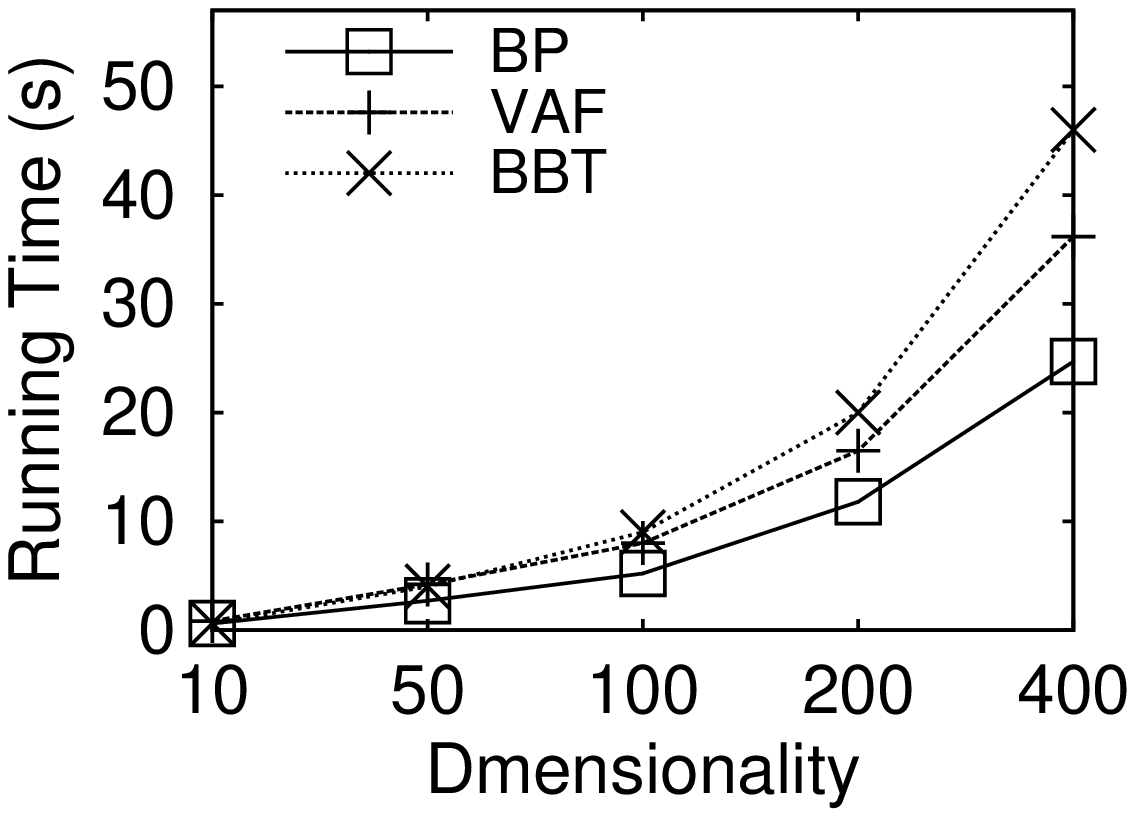}}
\caption{Impact of Dimensionality (Fonts)}
\label{Impact of Dimensionality}
\end{figure}

\begin{figure}
\hspace{-10pt}
\subfigure[I/O cost]{
\includegraphics[width=.25\textwidth]{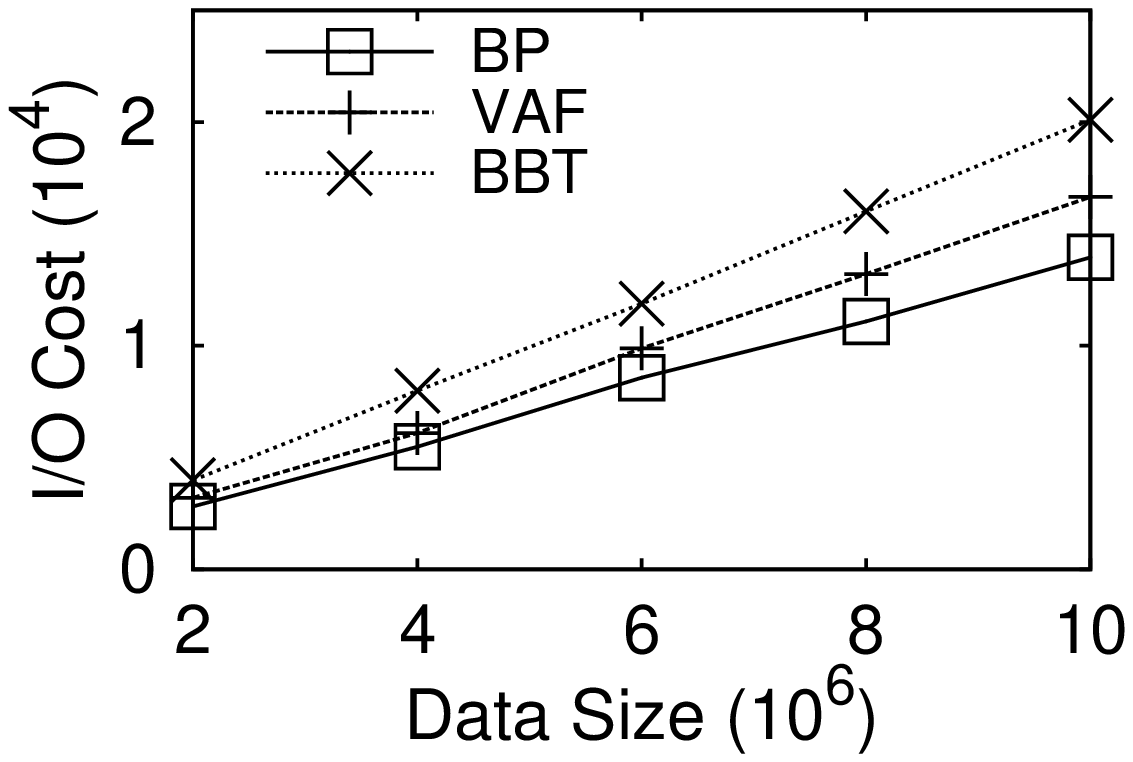}}
\hspace{-10pt}
\subfigure[Running time]{
\includegraphics[width=.25\textwidth]{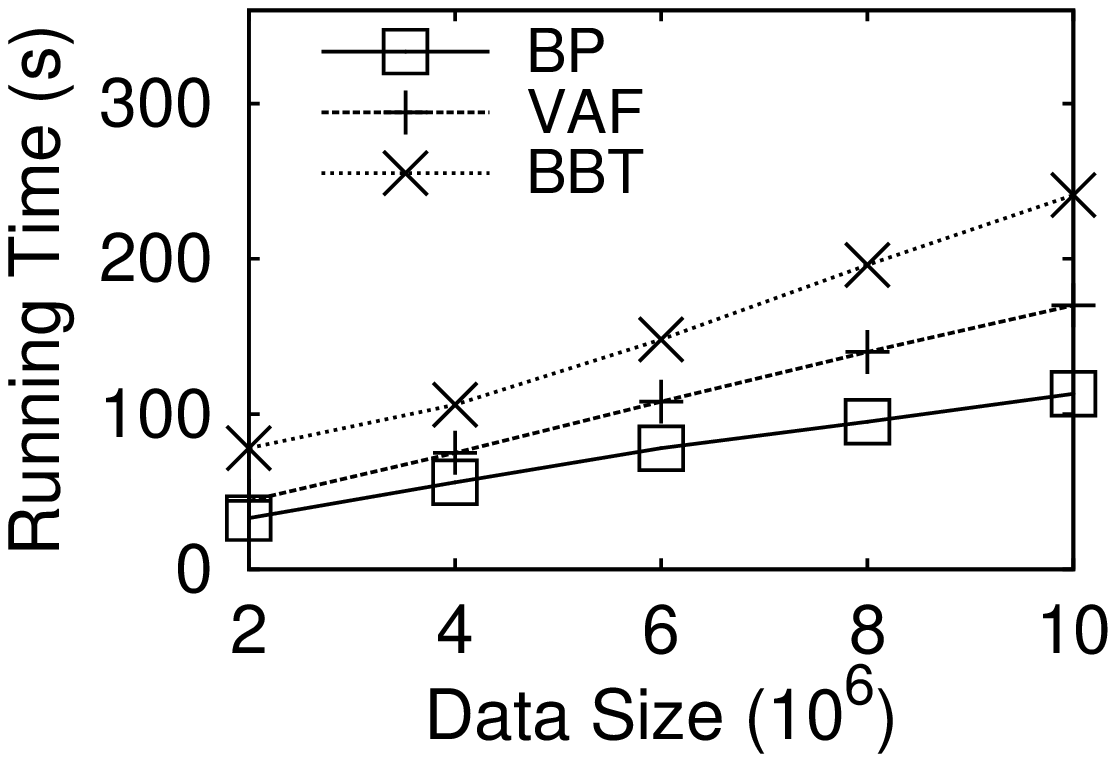}}
\caption{Impact of Data Size (Sift)}
\label{Impact of Data Size}
\end{figure}
\vspace{-10pt}
\subsection{Evaluations of Approximate Solution}\label{Approximation}
We evaluate our proposed approximate solution in this section. Our solution is compared with the state-of-the-art approximate solution proposed in~\cite{DBLP:conf/icml/CovielloMCL13} denoted as "Var" in our paper. We firstly define overall ratio denoted as $OR=\frac{1}{k}\sum_{i=1}^{k}{\frac{D_{f}(p_{i},q)}{D_{f}(p_{i}^{*},q)}}$, where $p_{i}$ is the $i$-th point returned and $p_{i}^{*}$ is the exact $i$-th NN point. Overall ratio describes the accuracy of the approximate solutions. A smaller overall ratio means a higher accuracy. Since "Var" is based on the traditional BB-tree, we also extend the memory-resident BB-tree to a disk-resident index structure following the idea of our proposed BB-forest to test the I/O cost. In addition, the running time is also tested for evaluating its efficiency. The results are shown in Fig.~\ref{EAS (Normal)}.

Fig.~\ref{EAS (Normal)}(a) shows the overall ratios by varying $k$ from 20 to 100. Generally, a larger $k$ can lead to a larger overall ratio. Moreover, we evaluate the overall ratio when the probability guarantee is set to 0.7, 0.8 and 0.9, respectively. From the experimental results, the overall ratio decreases as $p$ increases which indicates that a higher probability guarantee leads to a higher accuracy. Compared with "Var", our solution almost performs better in all cases on Normal. Similarly, we evaluate the I/O cost and the running time by varying $k$ and $p$. The experimental results are shown in Figs.~\ref{EAS (Normal)}(b) and~\ref{EAS (Normal)}(c). Generally, the I/O cost and the running time increase as $k$ increases. And we observe reverse trends from the results when varying $p$ because the searching range will be extended as $p$ increases. Since "Var" reduce the number of nodes checked in the searching process using data's distributions, the I/O cost and the running time are reduced. Even though, the I/O cost and the running time of "Var" are larger than our solution in most cases. It indicates that our approximate solution can yield the higher efficiency while ensuring the accuracy. The experimental results on Uniform are similar to those on Normal, they are shown in Section 6 in the supplementary file.
\begin{figure*}
\centering
\hspace{-10pt}
\subfigure[Overall Ratio]{
\includegraphics[width=0.25\textwidth]{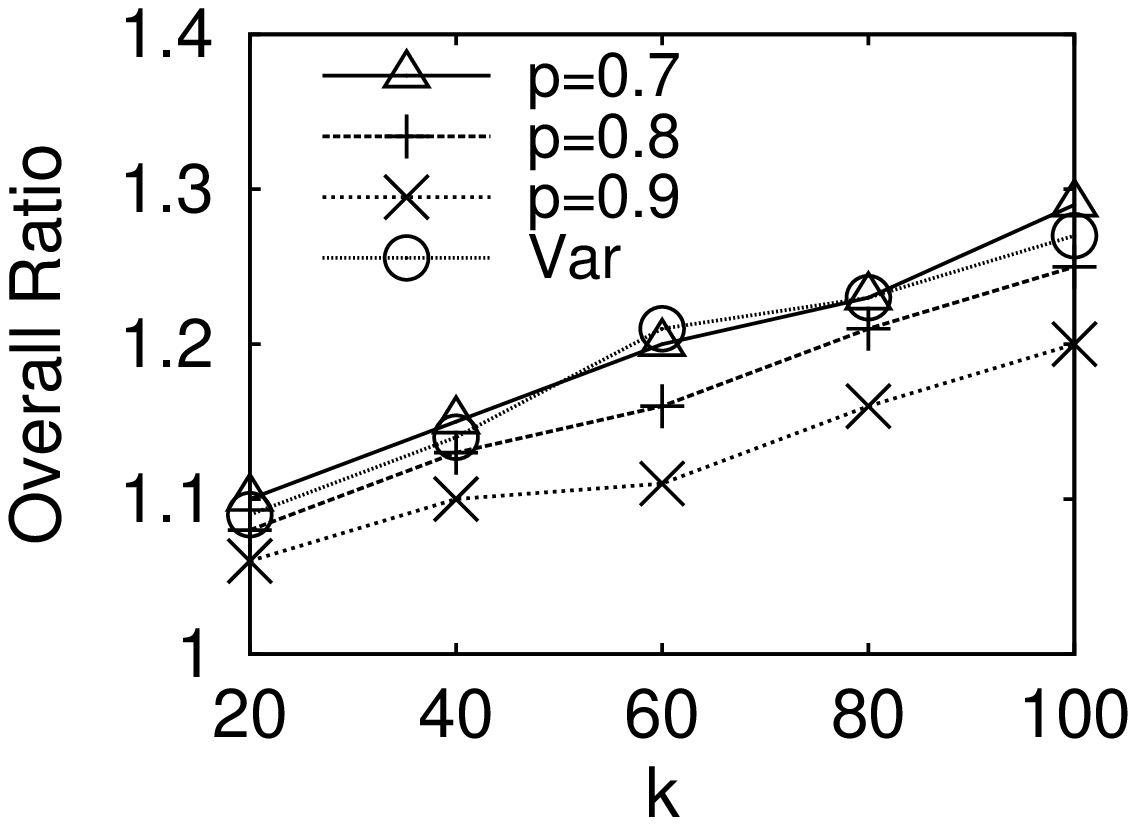}}
\hspace{-10pt}
\subfigure[I/O Cost]{
\includegraphics[width=0.25\textwidth]{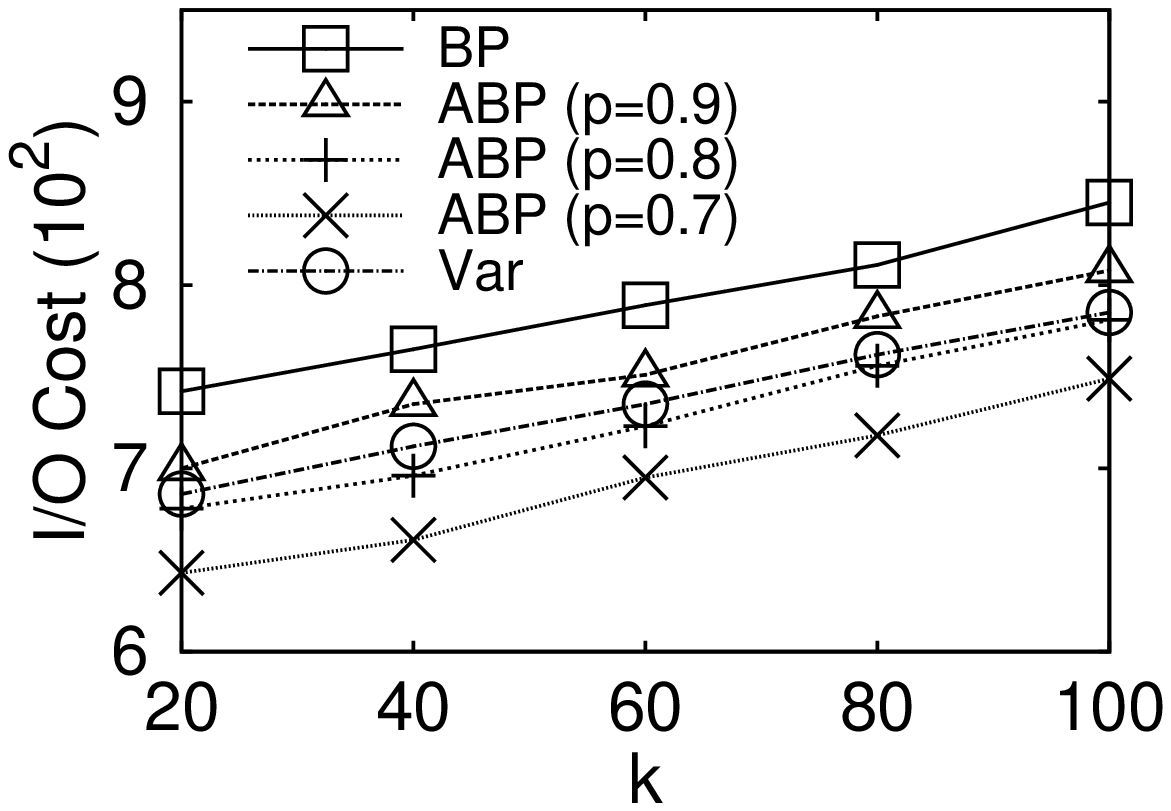}}
\hspace{-10pt}
\subfigure[Running Time]{
\includegraphics[width=0.25\textwidth]{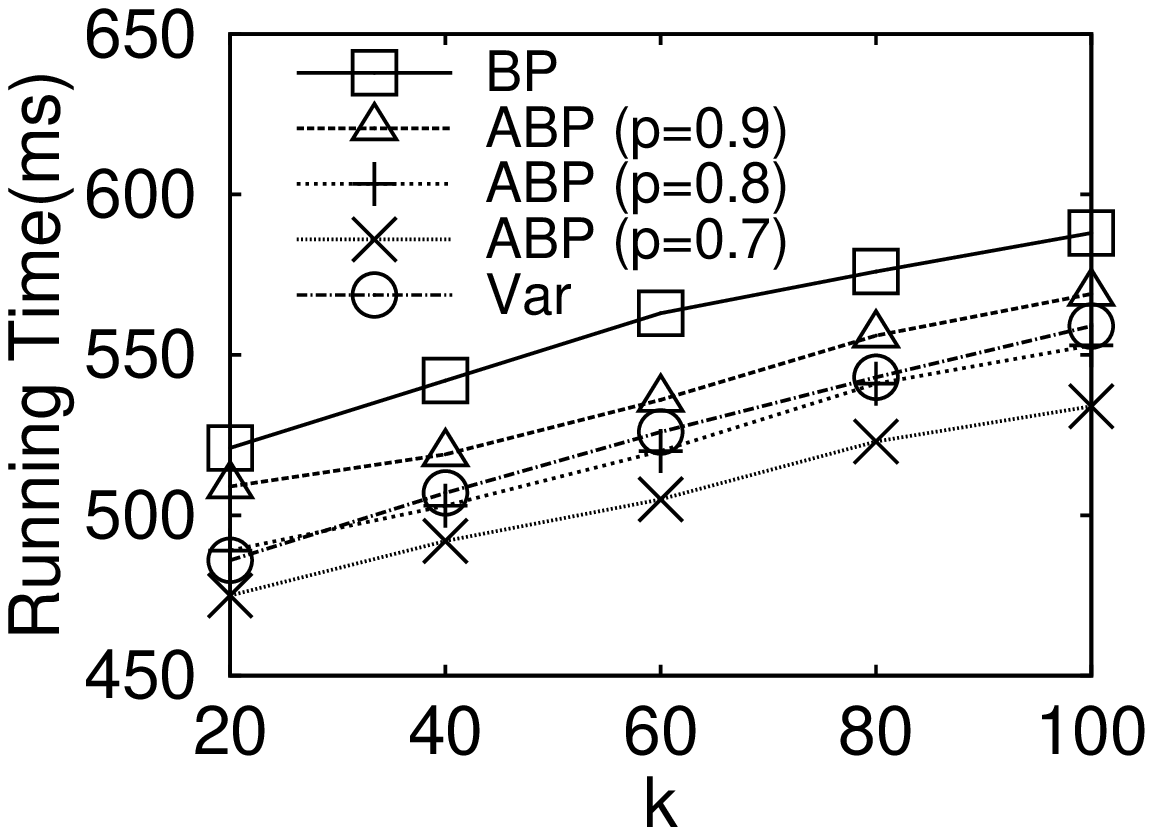}}
\caption{Evaluations of Approximate Solution (Normal)}
\label{EAS (Normal)}
\end{figure*}
\vspace{-10pt}
\section{Conclusion}\label{Conclusion}
In this paper, we address the important problem of high-dimensional \emph{k}NN search with Bregman distances and propose a dimensionality partitioning approach named BrePartition. BrePartition follows a partition-filter-refinement framework. We have proposed a number of novel techniques to overcome the challenges of the problem. First, we derive an effective upper bound based on Cauchy inequality as the pruning condition to significantly reduce the number of candidates we need to check closely. Second, we optimize the dimensionality partitioning by computing the optimized number of partitions to reduce the running time and devising a strategy called PCCP to further reduce the size of the candidate set. Third, we design an integrated index structure, named BB-forest, which consists of BB-trees for all the individual subspaces. In addition, we extend our exact solution to an approximate version via data's distribution. Experimental results demonstrate that our method can yield significant performance improvement in CPU's running time and I/O cost. In the future work, we will further improve the existing brief approximate solution and propose a more efficient solution by converting the Bregman distance into Euclidian distance and employing traditional metric searching methods to solve the high-dimensional \emph{k}NN search with Bregman distances. In addition, we will also improve our designed BB-forest so that it can support inserting or deleting large-scale data more efficiently.
\vspace{-10pt}
\ifCLASSOPTIONcompsoc
\bibliographystyle{IEEEtran}
\bibliography{reference}
\ifCLASSOPTIONcaptionsoff
  \newpage
\fi



%

%
\vspace{-40pt}
\begin{IEEEbiography}[{\includegraphics[width=1in,height=1.25in,clip,keepaspectratio]{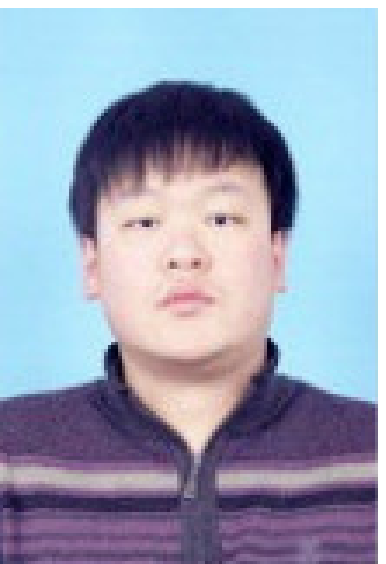}}]{Yang Song}
received the BE degree in automation from Northeastern University, China, in 2016. He is currently working toward the Ph.D degree in computer application technology at Northeastern University, China. His research interests include query processing and query optimization.
\end{IEEEbiography}
\vspace{-40pt}
\begin{IEEEbiography}[{\includegraphics[width=1in,height=1.25in,clip,keepaspectratio]{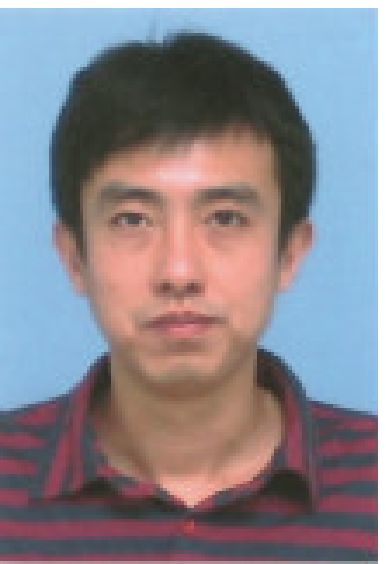}}]{Yu Gu}
received his Ph.D. degree in computer software and theory from Northeastern University, China, in 2010. Currently, he is a professor at Northeastern University, China. His current research interests include big data processing, spatial data management and graph data management. He is a senior member of China Computer Federation (CCF).
\end{IEEEbiography}
\vspace{-40pt}
\begin{IEEEbiography}[{\includegraphics[width=1in,height=1.25in,clip,keepaspectratio]{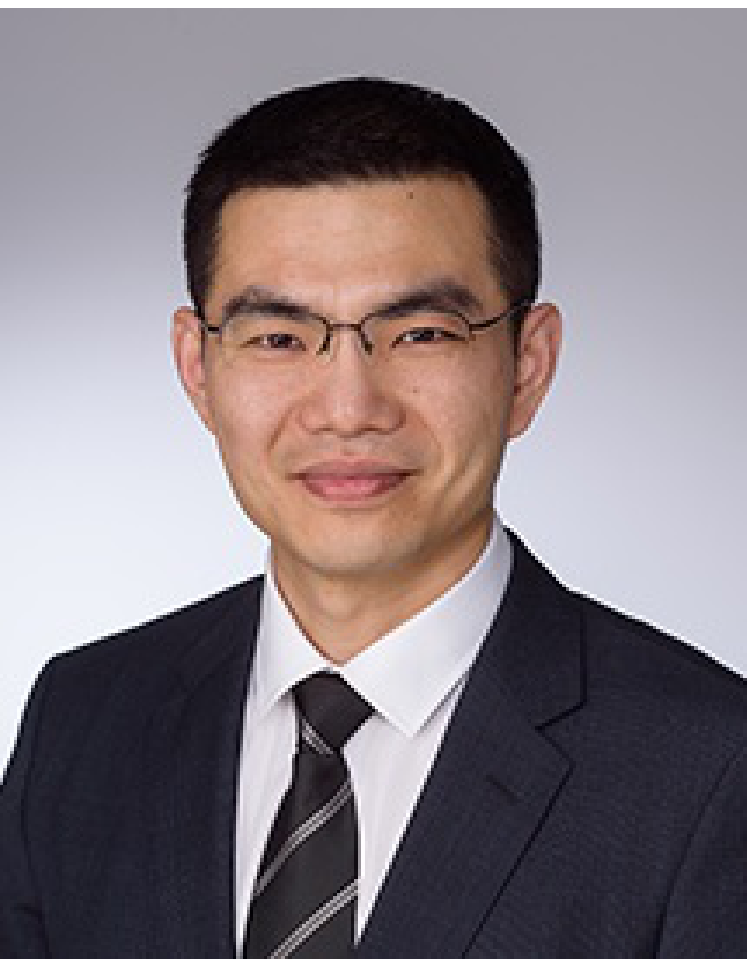}}]{Rui Zhang}
is a Professor at the School of Computing and Information Systems of the University of Melbourne. His research interests include big data and machine learning, particularly in spatial and temporal data analytics, database indexing, chatbots and recommender systems. Professor Zhang has won several awards including the Future Fellowship by the Australian Research Council in 2012, Chris Wallace Award for Outstanding Research by the Computing Research and Education Association of Australasia (CORE) in 2015, and Google Faculty Research Award in 2017.
\end{IEEEbiography}
\vspace{-40pt}
\begin{IEEEbiography}[{\includegraphics[width=1in,height=1.25in,clip,keepaspectratio]{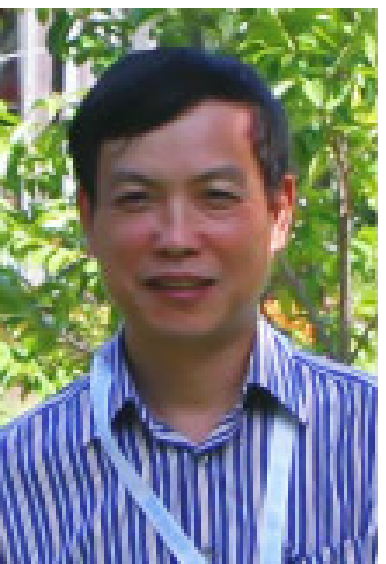}}]{Ge Yu}
received his PH.D. degree in computer science from Kyushu University of Japan in 1996. He is currently a professor at Northeastern University of China. His research interests include distributed and parallel database, OLAP and data warehousing, data integration, graph data management, etc. He has published more than 200 papers in refereed journals and conferences. He is a fellow of CCF and a member of the IEEE Computer Society, IEEE, and ACM.
\end{IEEEbiography}




\end{document}